\documentclass[12pt]{article}
\usepackage[english]{babel}
\usepackage{amsthm,amssymb,amsfonts,amsmath,amscd}
\usepackage{graphicx}
\usepackage{nccmath}
\usepackage{bm}

\renewcommand{\thesection}{\arabic{section}}
\makeatletter \@addtoreset{equation}{section} \makeatother
\newtheorem{proposition}{Proposition}
\newtheorem{theorem}{Theorem}
\newtheorem{lemma}{Lemma}
\newtheorem{remark}{Remark}
\renewcommand{\theequation}{\thesection.\arabic{equation}}

\def\mdet{\mathrm{det}}
\def\intd{\displaystyle\int}
\def\Tr{\mathrm{Tr}\,}

\def\Item$#1${\item $\displaystyle#1$
   \hfill\refstepcounter{equation}(\theequation)}

\begin{document}


\title{Universality of the second mixed moment of the characteristic polynomials of the 1D band matrices:
real symmetric case}

\author{Tatyana Shcherbina
 \\ \\
 \small{Chebyshev Laboratory, Mathematical Department of St.Petersburg State University,}\\
\small{14th Line, 29b, Saint Petersburg, 199178 Russia}
\\ \small{tshcherbi@gmail.com}
}

\date{\today}

\maketitle
\begin{abstract}
We prove that the asymptotic behavior of the second mixed moment of the characteristic polynomials of the
$N\times N$ 1D Gaussian real symmetric band matrices with the width of the band $W\gg N^{1/2}$ coincides with those for the Gaussian
Orthogonal Ensemble (GOE). Here we adapt the approach of \cite{TSh:14}, where 
 the Hermitian case was considered.
\end{abstract}

\section{Introduction}
The paper is the continuation of  \cite{TSh:14} to which we will frequently refer in this paper.
In \cite{TSh:14} we proved that the asymptotic behavior of the second mixed moment of the characteristic polynomials of the
1D Gaussian Hermitian random band matrices with $W\gg N^{1/2}$ coincides with those for the Hermitian
random matrices with i.i.d. (modulo symmetry) Gaussian random entries (GUE). The convenient
integral representation for the second correlation function of the characteristic polynomials
was obtained there by using the supersymmetry techniques (SUSY).
The SUSY method is widely used in the physics literature (see, e.g., \cite{Ef, M:00}) and is
potentially very powerful but the rigorous
control of the
integral representations, which can be obtained by this method, is difficult.
So far the most of rigorous results obtained by using the SUSY approach concern the case of  unitary symmetry. 
The goal of this paper is to show that the SUSY approach can be applied to  the case of the orthogonal symmetry as well, as to the unitary one.

We consider the real symmetric $N\times N$ matrices $H_N$ 
 (we enumerate indices of entries by $i,j\in \mathcal{L}$, where
$\mathcal{L}=[-n,n]^d\cap \mathbb{Z}^d$, $N=(2n+1)^d$) whose entries $H_{ij}$ are random
real Gaussian variables with mean zero such that
\begin{equation}\label{ban}
\mathbf{E}\big\{ H_{ij}H_{lk}\big\}=(\delta_{ik}\delta_{jl}+\delta_{il}\delta_{jk})J_{ij},\quad -n\le i, j, k, l\le n,
\end{equation}
where
\begin{equation}\label{J}
J_{ij}=\left(-W^2\Delta+1\right)^{-1}_{ij},
\end{equation}
and $\Delta$ is the discrete Laplacian on $\mathcal{L}$ with Neumann boundary conditions (cf. \cite{TSh:14}, eq. (1.1) -- (1.2)).
Note that  the variance of matrix elements $J_{ij}$ is exponentially small when $|i-j|\gg W$, and so  $W$ can be considered as the width of the band.
In this paper we will focus on the one-dimensional case ($d=1$).

The probability law of real symmetric 1D RBM  can be written in the form
\begin{equation}\label{band}
P_N(d H_N)=\prod\limits_{-n\le i<j\le n}\dfrac{dH_{ij}}{\sqrt{2\pi J_{ij}}}e^{-\frac{H_{ij}^2}{2J_{ij}}}
\prod\limits_{i=-n}^n\dfrac{dH_{ii}}{\sqrt{4\pi J_{ii}}}e^{-\frac{H_{ii}^2}{4J_{ii}}}.
\end{equation}
Let $\lambda_1^{(N)},\ldots,\lambda_N^{(N)}$ be the eigenvalues of
$H_N$. Define their Normalized Counting Measure
(NCM) as
\begin{equation} \label{NCM}
\mathcal{N}_N(\sigma)=\sharp\{\lambda_j^{(N)}\in
\sigma,j=1,\ldots,N \}/N,\quad \mathcal{N}_N(\mathbb{R})=1,
\end{equation}
where $\sigma$ is an arbitrary interval of the real axis.
The behavior of $\mathcal{N}_N$ as $N\to\infty$ was studied for many ensembles.
For 1D RBM it was shown
in \cite{BMP:91, MPK:92} that $\mathcal{N}_{N}$ converges weakly, as $N,W\to\infty$, to a non-random measure
$\mathcal{\mathcal{N}}$, which is called the limiting NCM of the ensemble. The measure $\mathcal{N}$ is absolutely continuous
and its density $\rho$ is given by the well-known Wigner semicircle law :
\begin{equation}\label{rho_sc}
\rho(\lambda)=\dfrac{1}{2\pi}\sqrt{4-\lambda^2},\quad \lambda\in[-2,2].
\end{equation}
Random band matrices (RBM) are natural intermediate models
to study eigenvalue statistics
and quantum propagation in disordered systems as they interpolate between mean-field Wigner
matrices (hermitian or real symmetric matrices with i.i.d.  random entries) and random Schr$\ddot{\hbox{o}}$dinger operators, 
where only a random one-site potential
is present in addition to the Laplacian on a regular box in $\mathbb{Z}^d$. In particular, 
RBM can be used to model the Anderson metal-insulator phase transition.

Let $\ell$ be the localization length,
which describes the typical length scale of the eigenvectors of random matrices.
The system is called delocalized if $\ell$ is
comparable with the matrix size, and it is called localized otherwise.
Delocalized systems correspond to electric
conductors, and localized systems are insulators.

According to the physical conjecture (see \cite{Ca-Co:90, FM:91}) for 1D RBM the expected order of $\ell$ is
 $W^2$ (for the energy in the bulk of the spectrum), which means that varying $W$
we can see the crossover: for $W\gg \sqrt{N}$ the eigenvectors are expected to be
delocalized and for $W\ll \sqrt{N}$ they are localized. At the present time only some upper and lower
bounds for $\ell$ are proven rigorously. It is known from the paper \cite{S:09} that $\ell\le W^8$.
On the other side, in the  papers \cite{EK:11, Yau:12} it was proven first that $\ell\gg W^{7/6}$,
and then that $\ell\gg W^{5/4}$.

The questions of the order of the localization length are closely related to the universality conjecture
of the bulk local regime of the random matrix theory (see \cite{TSh:14} for more details).
%
In this language the conjecture about the crossover for real symmetric 1D RBM states that we get the same local eigenvalue statistics in the bulk of the spectrum
as for GOE 
(real symmetric matrices with i.i.d Gaussian entries) for $W\gg \sqrt{N}$ (which corresponds to delocalized states), and
we get another behavior, which is determined by the Poisson statistics, for $W\ll \sqrt{N}$
(and corresponds to localized states). For the general real symmetric Wigner matrices (i.e., $W=N$) the
bulk universality has been proved in \cite{EYY:10}, \cite{TV:11}. However, in the general case
of RBM the question of bulk universality of local spectral statistics
is still open even for $d=1$.

In this paper we consider  the
correlation functions (or the mixed moments) of characteristic polynomials, which can be defined as
\begin{equation}\label{F}
F_{2k}(\Lambda)=\intd \prod\limits_{s=1}^{2k}\mdet(\lambda_s-H_N)P_n(d\,H_N),
\end{equation}
where $P_n(d\,H_N)$ is defined in (\ref{band}),
and $\Lambda=\hbox{diag}\,\{\lambda_1,\ldots,\lambda_{2k}\}$ are real or complex parameters
that may depend on $N$.  Although $F_{2k}(\Lambda)$ is not a local object, it is also expected 
to be universal in some sense. Moreover, correlation functions of characteristic polynomials are expected to exhibit a crossover which is similar to that
of local eigenvalue statistic. In particular, for  the real symmetric 1D RBM they are expected to have the same local behavior 
as for GOE for $W\gg \sqrt{N}$, and the different behavior for $W\ll \sqrt{N}$.

As was mentioned before, an additional source of motivation for the current work is the development
of the SUSY approach in the context of random operators with non-trivial spatial structures.
Although in the case of RBM (and some related types of the Wegner models) the SUSY method has
been applied rigorously so far mostly to the density of states (see \cite{Con:87}, \cite{DPS:02}),
the result of \cite{TSh:14_1} for the second correlation function of the block-band matrices
gives hope that the method can be applied also for $R_k$. From the SUSY point of view
characteristic polynomials correspond to the so-called fermionic sector of the supersymmetric full model,
which describes the correlation functions $R_k$. So the analysis of the local regime of correlation functions
of the characteristic polynomial is an important step
towards the proof of the universality of the correlation functions $R_k$ for the case of real symmetric
1D RBM.

The asymptotic local behavior in the bulk of the spectrum of the $2k$-point mixed moment
for GOE is known. It was proved for $k=1$ by Br$\acute{\hbox{e}}$zin and Hikami
\cite{Br-Hi:01}, who
used the SUSY approach, and for general $k$ by Borodin and Strahov \cite{BorSt:06}, who used different
techniques, that
\begin{equation*}
F_{2k}\left(\Lambda_0+\hat{\xi}/N\rho(\lambda_0)\right)
=C_N\dfrac{\hbox{Pf}\,
\big\{DS(\pi(\xi_i-\xi_j))
\big\}_{i,j=1}^{2k}}{\triangle(\xi_1,\ldots,\xi_{2k})}(1+o(1)),
\end{equation*}
where
\begin{equation}\label{dS}
DS(x)=-\dfrac{3}{x}\dfrac{d}{dx}\dfrac{\sin x}{x}=3\Big(\dfrac{\sin x}{x^3}-
\dfrac{\cos x}{x^2}\Big),
\end{equation}
$\triangle(\xi_1,\ldots,\xi_k)$ is the
Vandermonde determinant of $\xi_1,\dots, \xi_k$, and $$\hat{\xi}=\hbox{diag}\,\{\xi_1,\ldots,\xi_{2k}\},\quad
\Lambda_0=\lambda_0\cdot I.$$
In particular, for $k=1$ we have
\begin{equation*}
F_{2}\left(\Lambda_0+\hat{\xi}/N\rho(\lambda_0)\right)
=C_N\Big(\dfrac{\sin (\pi (\xi_1-\xi_2))}{\pi^3 (\xi_1-\xi_2)^3}-
\dfrac{\cos (\pi (\xi_1-\xi_2))}{\pi^2 (\xi_1-\xi_2)^2}\Big)(1+o(1)),
\end{equation*}
The last formula was proved also  for real symmetric Wigner and general sample covariance matrices
(see \cite{Kos:09}).

In this paper we obtain the same result for $k=1$ for matrices (\ref{ban})
as $N,W\to\infty$, $W^2=N^{1+\theta}$, $0<\theta\le 1$ (i.e., $W\gg \sqrt{N}$).

Set
\begin{equation*}
\lambda_j=\lambda_0+\dfrac{\xi_j}{N\rho(\lambda_0)},\quad j=1,2,
\end{equation*}
where $N=2n+1$, $\lambda_0\in (-2,2)$, $\rho$ is defined in (\ref{rho_sc}), and
$\{\xi_1,\xi_2\}$ are real parameters varying in any compact
set $K\subset \mathbb{R}$, and define
\begin{equation}\label{D_2}
D_2=\prod\limits_{l=1}^2F_2^{1/2}\Big(\lambda_0+\dfrac{\xi_l}{N\rho(\lambda_0)},\lambda_0+
\dfrac{\xi_l}{N\rho(\lambda_0)}\Big).
\end{equation}

 The main result of the paper is the following theorem :
\begin{theorem}\label{thm:1}
Consider the random matrices (\ref{ban}) -- (\ref{band}) with $W^2=N^{1+\theta}$, where $0<\theta\le 1$.
Define the second mixed moment $F_2$ of the characteristic polynomials as in (\ref{F}). Then we have
\begin{equation}\label{lim1}
\lim\limits_{N\to\infty}
D_2^{-1}F_{2}\Big(\Lambda_0+\hat{\xi}/(N\rho(\lambda_0))\Big)
=3\Big(\dfrac{\sin (\pi (\xi_1-\xi_2))}{\pi^3 (\xi_1-\xi_2)^3}-
\dfrac{\cos (\pi (\xi_1-\xi_2))}{\pi^2 (\xi_1-\xi_2)^2}\Big),
\end{equation}
and the limit is uniform in $\xi_1, \xi_2$ varying in any compact set $K\subset\mathbb{R}$. Here
$\rho(\lambda)$ and $D_2$ are defined in (\ref{rho_sc}) and (\ref{D_2}),
$\Lambda_0=\mathrm{diag}\,\{\lambda_0,\lambda_0\}$,
$\lambda_0\in (-2,2)$, $\hat{\xi}=\mathrm{diag}\,\{\xi_1,\xi_2\}$.
\end{theorem}
Theorem \ref{thm:1} is similar to the main Theorem 1 of \cite{TSh:14}.

The paper is organized as follows. In Sec. $2$ we obtain a convenient integral
representation for $F_{2}$,
using the integration
over the Grassmann variables.
In Sec. $3$ we give the sketch of the proof of Theorem
\ref{thm:1}. Sec. $4$  repeats some auxiliary results of \cite{TSh:14} needed for the proof.
In Sec. $5$ we prove Theorem \ref{thm:1}, applying the steepest descent method
to the integral representation. Sec. 6 is devoted to the proofs of the auxiliary statements.

\subsection{Notation}
We denote by $C$, $C_1$, etc. various $W$ and $N$-independent quantities below, which
can be different in different formulas. Integrals
without limits denote the integration (or the multiple integration) over the whole
real axis, or over the Grassmann variables.

Moreover,
\begin{itemize}
\item $N=2n+1;$

\item $\mathbf{E}\big\{\ldots\big\}$ is an expectation with respect to the measure (\ref{band});

\item $U_\varepsilon(x)=(x-\varepsilon,x+\varepsilon)\subset \mathbb{R};$

\Item
$\hfill
 a_{\pm}=\pm\dfrac{\sqrt{4-\lambda_0^2}}{2}=\pm\pi\rho(\lambda_0),\quad \overline{a}_\pm
=(a_\pm,\ldots,a_\pm)\in \mathbb{R}^N,\hfill
$\label{a_pm}

where $\rho$ is defined in (\ref{rho_sc});

\Item
 $
\hfill \sigma=\left(\begin{array}{cc}
0&1\\
-1&0
\end{array}\right),\quad
\sigma'=\left(\begin{array}{cc}
0&1\\
1& 0
\end{array}\right); \hfill
$\label{s-s'}

\item
$\hfill
\Lambda_0=\left(\begin{array}{cc}
\lambda_0&0\\
0& \lambda_0
\end{array}\right),\quad
\Lambda=\left(\begin{array}{cc}
\lambda_1&0\\
0& \lambda_2
\end{array}\right),\quad
\hat{\xi}=\left(\begin{array}{cc}
\xi_1&0\\
0& \xi_2
\end{array}\right),\quad
L=\left(\begin{array}{cc}
a_+&0\\
0& a_-
\end{array}\right); \hfill
$
\Item
$ \hfill
\Lambda_{0,4}=\left(\begin{array}{cc}
\Lambda_0&0\\
0& \Lambda_0
\end{array}\right),\quad
\hat{\xi}_4=\left(\begin{array}{cc}
\hat{\xi}&0\\
0& \hat{\xi}
\end{array}\right),\quad
L_4=\left(\begin{array}{cc}
L&0\\
0& L
\end{array}\right); \hfill
$
\label{xi_4}

\item $\mathring{U}(2)=U(2)/\big(U(1)\times U(1)\big)$, \quad $\mathring{Sp}(2)=Sp(2)/\big(Sp(1)\times Sp(1)\big)$ 

\item $d\mu$ is the Haar measure on $\mathring{U}(2)$, $d\nu$ is the Haar measure on $\mathring{Sp}(2)$;

\Item
$
 f(x)=(x+i\lambda_0/2)^2/2-\log(x-i\lambda_0/2)$\label{f}
 
$f_*(x)=\Re(f(x)-f(a_\pm))=\big(x^2-\lambda_0^2/4-\log(x^2+\lambda_0^2/4)\big)/2- \Re f(a_\pm);
$

\item $\Omega_\delta$ is a union of
\begin{align} \label{Om_delta}
 \Omega^+_\delta&=\{ \{a_j\}, \{b_j\}: a_j, b_j\in U_\delta(a_+) \,\,\forall j\},\\ \notag
 \Omega^-_\delta&=\{ \{a_j\}, \{b_j\}: a_j, b_j\in U_\delta(a_-)\,\,\forall j\},\\ \notag
\Omega^\pm_\delta&=\{ \{a_j\}, \{b_j\}: (a_j\in U_\delta(a_+),\,\, b_j\in U_\delta(a_-))\\ \notag
&\quad\quad\quad\quad\quad\quad\quad\quad\quad\quad\quad\quad\quad\,\,\hbox{or}\,\,
(a_j\in U_\delta(a_-),\, b_j\in U_\delta(a_+))\,\,\forall j\},
\end{align}
where $\delta=W^{-\kappa}$ and $\kappa<\theta/8$.

\Item
$ \hfill
 c_\pm=1-\dfrac{\lambda_0^2}{4}\pm
\dfrac{i\lambda_0}{2}\cdot \sqrt{1-\lambda_0^2/4},\quad c_0=\Re f(a_+)=\dfrac{2-\lambda_0^2}{4};\hfill
$\label{c_pm}

\Item
$\hfill
\mu_{\gamma}(x)=\exp\big\{-\frac{1}{2}\sum\limits_{j=-n+1}^{n}(x_j-x_{j-1})^2
-\frac{\gamma}{W^2}\sum\limits_{j=-n}^{n}x_j^2\big\};
\hfill$\label{mu}

\Item
$
\langle \ldots \rangle_0=Z_{\delta,\gamma}^{-1}\intd_{-\delta W}^{\delta W} (\ldots) \cdot
\mu_{\gamma}(x)\prod\limits_{q=-n}^nd x_q,\quad \quad Z_{\delta,\gamma}=
\intd_{-\delta W}^{\delta W} \mu_{\gamma}(x) \prod\limits_{q=-n}^nd x_q,$\label{angle}

$\langle \ldots \rangle=Z^{-1}_\gamma\intd (\ldots) \cdot \mu_{\gamma}(x) \prod\limits_{q=-n}^nd x_q, \quad\quad \quad\,\,\, Z_\gamma=
\intd \mu_{\gamma}(x) \prod\limits_{q=-n}^nd x_q,$

where $\delta>0$ and $\gamma\in \mathbb{C}$, $\Re\gamma>0$;

\item $\langle \ldots \rangle_*$
(and $\langle \ldots \rangle_{0,*}$) is (\ref{angle}) with $\mu_{\Re \gamma}(x)$
instead of $\mu_{\gamma}(x)$.
\end{itemize}

\section{Integral representation}
In this section we obtain an integral representation for $F_{2}$ of (\ref{F}) by using rather
standard SUSY techniques, i.e., integrals
over the Grassmann variables. Integration over the Grassmann variables has been introduced by Berezin and is widely used in the physics
literature.  A brief outline of the techniques can be found in \cite{TSh:14}, Sec. 2.1.

The main result of the section is the following proposition
\begin{proposition}\label{p:int_repr}
The second correlation function of the characteristic polynomials for 1D real symmetric Gaussian band
matrices, defined in (\ref{F}), can be represented as follows:
\begin{align}\label{F_rep}
&F_2\Big(\Lambda_0+\dfrac{\hat{\xi}}{N\rho(\lambda_0)}\Big)=-(2\pi^3)^{-N}\mdet^{-3} J
\int\exp\Big\{-\frac{W^2}{4}
\sum\limits_{j=-n+1}^n\Tr\,(F_j-F_{j-1})^2\Big\}\\ \notag
&\times\exp\Big\{-\frac{1}{4}\sum\limits_{j=-n}^n\Tr \Big(F_j+\dfrac{i\Lambda_{0,4}}{2}+
\dfrac{i\hat{\xi}_4}{N\rho(\lambda_0)}\Big)^2\Big\}
\prod\limits_{j=-n}^n
\mdet^{1/2}\big(F_j-i\Lambda_{0,4}/2\big)\prod\limits_{j=-n}^ndF_j,
\end{align}
where $\Lambda_{0,4}$ and $\hat{\xi}_4$ are defined in (\ref{xi_4}), and
\begin{equation}\label{F_j}
F_j=\left(
\begin{array}{llll}
x_j&w_{j1}&0&w_{j2}\\
\overline{w}_{j1}&y_j&-w_{j2}&0\\
0&-\overline{w}_{j2}&x_j&\overline{w}_{j1}\\
\overline{w}_{j2}&0&w_{j1}&y_j
\end{array}
\right),\quad dF_j=dx_j\, dy_j\, d\Re w_{j1}\, d\Im w_{j1}\, d\Re w_{j2}\, d\Im w_{j2}.
\end{equation}
Moreover, (\ref{F_rep}) can be rewritten in the form
\begin{align}\notag
&F_2\Big(\Lambda_0+\dfrac{\hat{\xi}}{N\rho(\lambda_0)}\Big)=-\dfrac{C(\xi)\mdet^{-3}J}{(24\pi)^{N}}
\intd\limits\exp\Big\{-\frac{W^2}{4}\sum\limits_{j=-n+1}^n\Tr (Q_j^*A_{j,4}Q_j-A_{j-1,4})^2\Big\}\\ \notag
&\times \exp\Big\{-\sum\limits_{j=-n}^n(f(a_j)+f(b_j))-
\frac{i}{2N\rho(\lambda_0)}\sum\limits_{j=-n}^n\Tr \big(R_jP_{-n}\big)^*A_{j,4}\,
(R_jP_{-n}\big)\hat{\xi}_4\Big\}
\\ \label{F_0_1}
&\times\prod\limits_{l=-n}^n(a_l-b_l)^4d\,\nu(P_{-n})\,d\overline{a\vphantom{b}}\, d\overline{b}\,
\prod\limits_{p=-n+1}^nd\nu(Q_p),
\end{align}
where $f$ is defined in (\ref{f}),
$A_{j,4}=\mathrm{diag}\{a_j,b_j,a_j,b_j\}$, $\{R_j\}$ and $P_{-n}$ are $4\times 4$ symplectic matrices,
$d\nu(P)$ is the Haar measure on $\mathring{Sp}(2)$, and
\begin{equation}\label{R_k}
R_k=\prod\limits_{s=k}^{-n+1}Q_s,\quad C(\xi)=\exp\Big\{\dfrac{\lambda_0(\xi_1+\xi_2)}{2\rho(\lambda_0)}+
\dfrac{\xi_1^2+\xi_2^2}{2N\rho(\lambda_0)^2}\Big\}.
\end{equation}
\end{proposition}
\begin{remark}
Formula (\ref{F_rep}) is valid for any dimension if we change the sum $\sum\Tr
(F_j-F_{j-1})^2$ to $\sum\Tr
(F_j-F_{j^\prime})^2$, where the last sum runs over all pairs of nearest neighbor
$j, j^\prime$ in the volume $\mathcal{L}\subset \mathbb{Z}^d$ (see the definition of RBM
(\ref{ban}) -- (\ref{J})).
\end{remark}
\begin{proof}
Representing determinants as integrals over Grassmann variables, we obtain
\begin{equation*}
\begin{array}{c}
F_2(\Lambda)={\bf E}\bigg\{\displaystyle\int
e^{-\sum\limits_{\alpha=1}^2\sum\limits_{j,k=-n}^n(\lambda_l-H_n)_{j,k}
\overline{\psi}_{j\alpha}\psi_{k\alpha}}\prod\limits_{\alpha=1}^2\prod\limits_{j=-n}^n
d\,\overline{\psi}_{j\alpha}d\,\psi_{j\alpha}\bigg\}\\
={\bf E}\bigg\{\displaystyle\int e^{-\sum\limits_{\alpha=1}^2\lambda_s\sum\limits_{j=-n}^n
\overline{\psi}_{j\alpha}\psi_{j\alpha}} \exp\bigg\{\sum\limits_{j<k}H_{jk}\sum\limits_{\alpha=1}^2\big(
\overline{\psi}_{j\alpha}\psi_{k\alpha}
+\overline{\psi}_{k\alpha}\psi_{j\alpha}\big)\\
+\sum\limits_{j=-n}^nH_{jj}\cdot\sum\limits_{\alpha=1}^2\overline{\psi}_{j\alpha}\psi_{j\alpha}\bigg\}
\prod\limits_{\alpha=1}^2\prod\limits_{j=-n}^nd\,\overline{\psi}_{j\alpha}d\,\psi_{j\alpha}\bigg\},
\end{array}
\end{equation*}
where $\{\psi_{j\alpha}\}$, $j=-n,\ldots,n$, $\alpha=1,2$ are the Grassmann variables ($2n+1$ variables for each
determinant in (\ref{F})). Here and below we use Greek letters such as $\alpha, \beta$ etc. for
the field index and Latin letters $j, k$ etc. for the position index.

Integrating over the measure (\ref{band}), we get
\begin{align}\label{usr}
&F_2(\Lambda)=\displaystyle\int \prod\limits_{\alpha=1}^2\prod\limits_{q=-n}^nd\,\overline{\psi}_{q\alpha}
d\,\psi_{q\alpha}
\exp\Big\{-\sum\limits_{\alpha=1}^2\lambda_\alpha\sum\limits_{p=-n}^n \overline{\psi}_{p\alpha}
\psi_{p\alpha}\Big\}\\ \notag\times
\exp&\Big\{\dfrac{1}{2}\sum\limits_{j<k}J_{jk}(\overline{\psi}_{j1}\psi_{k1}+\overline{\psi}_{j2}\psi_{k2}+
\overline{\psi}_{k1}\psi_{j1}+\overline{\psi}_{k2}\psi_{j2})^2+\sum\limits_{j=-n}^nJ_{jj}
(\overline{\psi}_{j1}\psi_{j1}+\overline{\psi}_{j2}\psi_{j2})^2\Big\}.
\end{align}
Now we will need the Hubbard-Stratonovich transform (see,  e.g., \cite{Sp:12}).
This is a well-known simple trick, which
is just the Gaussian integration. In the simplest form it looks as following:
\begin{equation}\label{Hub}
e^{a^2/2}=(2\pi)^{-1/2} \int e^{-x^2/2+ax}dx.
\end{equation}
Here $a$ can be complex number or the sum of the products of even numbers of Grassmann variables.

Applying a couple of times  (\ref{Hub}), we can write:
\begin{align}\label{HS1}
&\int \exp\Big\{-(J^{-1}x,x)/2+i\sum\limits_{j=-n}^n x_{j}
\overline{\psi}_{j1}\psi_{j1}\Big\} \prod\limits_{j=-n}^n dx_{j}\\ \notag
 &\quad =(2\pi)^{N/2}\cdot\mdet^{1/2} J\cdot
\exp\Big\{-\dfrac{1}{2}\sum\limits_{j,k=-n}^nJ_{jk}\overline{\psi}_{j1}\psi_{j1}
\overline{\psi}_{k1}\psi_{k1}\Big\},\\ \label{HS2}
&\int \exp\Big\{-(J^{-1}y,y)/2+i\sum\limits_{j=-n}^n y_{j}
\overline{\psi}_{j2}\psi_{j2}\Big\} \prod\limits_{j=-n}^n dy_{j}\\ \notag
 &\quad =(2\pi)^{N/2}\cdot\mdet^{1/2} J\cdot
\exp\Big\{-\dfrac{1}{2}\sum\limits_{j,k=-n}^nJ_{jk}\overline{\psi}_{j2}\psi_{j2}
\overline{\psi}_{k2}\psi_{k2}\Big\},\\ \notag
\end{align}
where $x=\{x_j\}_{j=-n}^n$,  $y=\{y_j\}_{j=-n}^n$. In addition,
\begin{align}\label{HS3}
&\int \exp\Big\{-(J^{-1}\Re w_1,\Re w_1)-(J^{-1}\Im w_1,
\Im w_1)\Big\}\\ \notag
&\quad\quad\times\exp\Big\{i\sum\limits_{j=-n}^n w_{j1}
\overline{\psi}_{j1}\psi_{j2}+i\sum\limits_{j=-n}^n \overline{w}_{j1}
\overline{\psi}_{j2}\psi_{j1}\Big\} \prod\limits_{q=-n}^n d\Re w_{q1}d\Im w_{q1}\\ \notag
 &\quad =\pi^N\cdot\mdet J\cdot
\exp\Big\{-\sum\limits_{j\ne k}J_{jk}\overline{\psi}_{j1}\psi_{j2}
\overline{\psi}_{k2}\psi_{k1}-\sum\limits_{j=-n}^nJ_{jj}\overline{\psi}_{j1}\psi_{j2}
\overline{\psi}_{j2}\psi_{j1}\Big\},
\end{align}
\begin{align}\label{HS4}
&\int \exp\Big\{-(J^{-1}\Re w_2,\Re w_2)-(J^{-1}\Im w_2,
\Im w_2)\Big\}\\ \notag
&\quad\quad\times \exp\Big\{i\sum\limits_{j=-n}^n w_{j2}
\overline{\psi}_{j1}\overline{\psi}_{j2}+i\sum\limits_{j=-n}^n \overline{w}_{j2}
\psi_{j1}\psi_{j2}\Big\} \prod\limits_{q=-n}^n d\Re w_{q2}d\Im w_{q2}\\ \notag
 &\quad =\pi^N\cdot\mdet J\cdot
\exp\Big\{-\sum\limits_{j\ne k}J_{jk}\overline{\psi}_{j1}\overline{\psi}_{j2}
\psi_{k1}\psi_{k2}-\sum\limits_{j=-n}^nJ_{jj}\overline{\psi}_{j1}\overline{\psi}_{j2}
\psi_{j1}\psi_{j2}\Big\},
\end{align}
where $\Re w_\alpha=\{\Re w_{j\alpha}\}_{j=-n}^n$, $\Im w_\alpha=\{\Im w_{j\alpha}\}_{j=-n}^n$, $\alpha=1,2$.

Substituting (\ref{HS1}) -- (\ref{HS4}) and (\ref{J}) for $J^{-1}_{jk}$ into (\ref{usr}), putting
$\Lambda=\Lambda_0+\hat{\xi}/N\rho(\lambda_0)$, and
integrating over the Grassmann variables, we obtain
\begin{align*}\notag
&F_2\Big(\Lambda_0+\dfrac{\hat{\xi}}{N\rho(\lambda_0)}\Big)=-(2\pi^3)^{-N}\mdet^{-3} J\int\exp\Big\{-\frac{W^2}{4}
\sum\limits_{j=-n+1}^n\hbox{Tr}\,(F_j-F_{j-1})^2\Big\}\\
&\times\exp\Big\{-\frac{1}{4}\sum\limits_{j=-n}^n\Tr F_j^{\,2}\Big\}
\prod\limits_{j=-n}^n
\mdet^{1/2}\big(F_j-i\Lambda_{0,4}-i\hat{\xi}_4/N\rho(\lambda_0)\big)\prod\limits_{j=-n}^ndF_j
\end{align*}
with $F_j$ of (\ref{F_j}) and $\Lambda_{0,4}$, $\hat{\xi}_4$ of (\ref{xi_4}). This gives (\ref{F_rep}) after shifting $F_j\to F_j+i\Lambda_{0,4}/2+i\hat{\xi}_4/N\rho(\lambda_0)$. The reason
of such a shift is that we need to have saddle-points lying on the contour of the integration
(see (\ref{a_pm})).

The matrices of the form (\ref{F_j}) have two eigenvalues $a_j,b_j$ of the multiplicity two and
can be considered as quaternion $2\times 2$ matrices.
In this language $F$ is a quaternion self-dual Hermitian  matrix, and it can be diagonalized
by the quaternion unitary $2\times 2$ matrices $Sp(2)$ (see , e.g., \cite{Me:91}, Chapter 2.4),
i.e., unitary $4\times 4$ matrices $P$ which admit the relation
\[
P\, \left(\begin{array}{cc}
0&I_2\\
-I_2&0
\end{array}\right)\, P^{t}=\left(\begin{array}{cc}
0&I_2\\
-I_2&0
\end{array}\right).
\]
Change the variables to $F_j=P_j^*A_{j,4} P_j$, where $P_j\in \mathring{Sp}(2)$ and
$A_{j,4}=\hbox{diag}\,\{a_j,b_j,a_j,b_j\}$.
Then $d F_j$ of (\ref{F_j}) becomes (see, e.g., \cite{Me:91} )
$$\dfrac{\pi^2}{12}(a_j-b_j)^4da_j\,db_j
d\nu(P_j),$$ where $d\nu(P_j)$ is the normalized to unity Haar measure on the symplectic group $\mathring{Sp}(2)$.

Thus, we have
\begin{align*}
F_2\Big(\Lambda_0+\dfrac{\hat{\xi}}{N\rho(\lambda_0)}\Big)=&-\dfrac{C(\xi)\mdet^{-3}J }{(24\pi)^{N}}
\int\,d\overline{a\vphantom{b}}\, d\overline{b}\int_{\mathring{Sp}(2)^N}\,\prod\limits_{j=-n}^nd\nu(P_j)\\
&\times\exp\Big\{-\frac{W^2}{4}\sum\limits_{j=-n+1}^n\Tr (P_j^*A_{j,4}P_j-
P_{j-1}^*A_{j-1,4}P_{j-1})^2\Big\}\\
&\times
\exp\Big\{-\frac{1}{4}\sum\limits_{j=-n}^n\Tr \Big(A_{j,4}+\frac{i\Lambda_{0,4}}{2}\Big)^2-
\dfrac{i}{2N\rho(\lambda_0)}\sum\limits_{j=-n}^n\Tr P_j^*A_{j,4}P_j\hat{\xi}_4\Big\}
\\
&\times\prod\limits_{k=-n}^n (a_k-i\lambda_0/2\big) (b_k-i\lambda_0/2\big)\prod\limits_{k=-n}^n
(a_k-b_k)^4,
\end{align*}
where $C(\xi)$ is defined in (\ref{R_k}), and
\begin{equation*}
d\overline{a\vphantom{b}}=\prod\limits_{j=-n}^nda_j,\quad d\overline{b}=\prod\limits_{j=-n}^ndb_j.
\end{equation*}
Now changing the ``angle variables'' $P_j$ to $Q_j=P_jP_{j-1}^*$, $j=-n+1,\ldots,n$ (i.e., the new variables are
$P_{-n}, Q_{-n+1}, Q_{-n+2},\ldots,Q_n$), we get (\ref{F_0_1}).
\end{proof}

\section{Sketch of the proof of Theorem \ref{thm:1}}

The strategy of the proof is the same as in \cite{TSh:14}. The main difference is that now we perform the integration over $\mathring{Sp}(2)$
instead of $\mathring{U}(2)$, which is much more complicated. 

\smallskip

So first we note that the main integrations over $a_j$, $b_j$ are the same as in \cite{TSh:14}, eq.(2.11),
and so the expected saddle-points for each $a_j$ and $b_j$
are still $a_\pm$ (see (\ref{a_pm})). Moreover, we can use the results of
\cite{TSh:14}, Sec. 4.1 -- 4.2, where the properties of the function $f$ and of the complex Gaussian
distribution $\mu_\gamma$ of (\ref{mu}) were studied (see Sec. 4.1).

The second step is to prove that the main contribution to the integral (\ref{F_0_1}) is given by the integral $\Sigma$ over $\Omega_\delta$ (see (\ref{Om_delta})). More precisely, we are going to prove that
\begin{equation}\label{F_2}
F_2\Big(\Lambda_0+\dfrac{\hat{\xi}}{N\rho(\lambda_0)}\Big)=-\dfrac{C(\xi) \mdet^{-3}J}{(24\pi)^{N}}
\cdot \Sigma\cdot(1+o(1)),\quad W\to \infty.
\end{equation}
The bound for the complement $|\Sigma_c|$ can be obtained by inserting the absolute value inside the integral and
by performing exactly the integral over the symplectic groups. After this, since we are
far from the saddle-points of $f$, one can control the integral in the same way as in
\cite{TSh:14} (see Lemma \ref{l:2}).

The next step is the calculation of $\Sigma$ (see Sec. 5.2, Lemma \ref{l:sigma}).
We are going to show that the main contribution to $\Sigma$ is given by $\Sigma_\pm$, i.e., the integral
over $\Omega_\delta^{\pm}$. Consider $\Omega_\delta^{\pm}$. First note that shifting
$$P_j \to \left(
\begin{array}{ll}
\sigma'&0\\
0&\sigma'
\end{array}
\right) P_j$$
for some $j$ ($\sigma'$ is defined in (\ref{s-s'})), we can rotate each domain of type $$\{ \{a_j\}, \{b_j\}: (a_j\in U_\delta(a_+),\,\,
b_j\in U_\delta(a_-))\,\,\hbox{or}\,\,
(a_j\in U_\delta(a_-),\, b_j\in U_\delta(a_+))\,\,\forall j\}$$ to the $\delta$-neighborhood
of the point $(\overline{a}_+,\overline{a}_-)$ with $\overline{a}_\pm$ of (\ref{a_pm}). Thus, we can
consider the contribution over $\Omega^{\pm}_{\delta}$ as
$2^{N}$ contributions of the $\delta$-neighborhood of the point $(\overline{a}_+,\overline{a}_-)$.
Consider this neighborhood,
and change the variables as
\begin{align*}
a_j&\to a_++\tilde{a}_j/W,\quad |\tilde{a}_j|\le\delta W,\\ \notag
b_j&\to a_-+\tilde{b}_j/W, \quad\, |\tilde{b}_j|\le\delta W,
\end{align*}
and set $\widetilde{A}_{j,4}=\hbox{diag}\,\{\tilde{a\vphantom{b}}_j,\tilde{b}_j, \tilde{a\vphantom{b}}_j,\tilde{b}_j\}$.
To compute $\Sigma_\pm$, one has to perform first the integral over the symplectic groups.
This integral is some analytic in $\{\tilde{a\vphantom{b}}_j/W\}$, $\{\tilde{b}_j/W\}$ function $\mathcal{F}$.
As in \cite{TSh:14}, the main idea is to prove that the leading part of this function can be obtained by
replacing all $Q_s$ in the ``bad'' term
$$\exp\Big\{-\frac{i}{2N\rho(\lambda_0)}\sum\limits_{j=-n}^n
\Tr \big(\prod\limits^{-n+1}_{s=j} Q_s\cdot P_{-n}\big)^*(L_4+\widetilde{A}_{s,4}/W)\,
(\prod\limits^{-n+1}_{s=j} Q_s \cdot P_{-n}\big)\hat{\xi}_4\Big\}$$
with $I$. To this end, we expand the ``bad'' term into a series and for each summand, which is
analytic in $\{\tilde{a\vphantom{b}}_j/W\}$, $\{\tilde{b}_j/W\}$, find the bound for
its Taylor coefficients (see Lemma \ref{l:un}). This means that we obtain the proper majorant for $\mathcal{F}$
in the sense of \cite{TSh:14} (i.e., some function whose Taylor expansion's coefficients are at least
the absolute value of the corresponding coefficient of
the Taylor expansion of $\mathcal{F}$), which helps to change the averaging over the complex measure by
the averaging of the majorant over the positive one (see Lemma \ref{l:maj}).
Then, similarly to \cite{TSh:14}, we will show that the leading term of $\Sigma_\pm$ is the integral over the Gaussian measures
$\mu_{c_\pm}$ in $\{a_j\}$ and $\{b_j\}$
variables, and the integral over the
symplectic group $d\nu(P_{-n})$ which gives the kernel (\ref{dS}).
This yields an asymptotic expression for $\Sigma_\pm$ (see Lemma \ref{l:sig}).

Also it will be shown in Sec. 5.2.2 that the integrals $\Sigma_+$ and $\Sigma_-$ over $\Omega_\delta^+$ and
$\Omega_\delta^-$ have smaller orders than $\Sigma_\pm$ (see Lemma \ref{l:sig_+}).

\section{Preliminary results}
In this section we restated the results of \cite{TSh:14}, Sec. 4.2., where the properties of the
complex Gaussian
distribution $\mu_\gamma$ of (\ref{mu}) were studied. All proofs can be found in \cite{TSh:14}.

First note that the straightforward calculation gives in the small neighborhood of $a_\pm$
\begin{equation}\label{f_exp}
f(x)-f(a_\pm)=c_{\pm}(x-a_\pm)^2+s_3(x-a_\pm)^3+\ldots=:c_{\pm}(x-a_\pm)^2+\varphi_\pm(x-a_\pm),
\end{equation}
where $c_\pm$ is defined in (\ref{c_pm}) and $|\varphi_\pm(x-a_\pm)|=O(|x-a_\pm|^3)$.

Now set
\begin{align}\label{mu_m}
\mu^{(m)}_{\gamma}(x)=\exp\big\{-\frac{1}{2}\sum\limits_{j=2}^{m}(x_j-x_{j-1})^2
-\frac{\gamma}{W^2}\sum\limits_{j=1}^{m}x_j^2\big\}.
\end{align}

\begin{lemma}[\cite{TSh:14}, Lemma 3]\label{l:okr}
We have for any $\gamma\in \mathbb{C}$, $\Re \gamma>0$

\begin{enumerate}
\item
 $\hfill 
 Z^{(m)}_\gamma:=\displaystyle\int
\mu^{(m)}_{\gamma}(x)
\prod\limits_{q=1}^md x_q=(2\pi)^{m/2}\mdet^{-1/2}(-\Delta+2\gamma/W^2)\hfill $

$ \hfill
=(2\pi)^{m/2}\Big(\dfrac{\sqrt{2\gamma}}{W}\sinh \dfrac{m\sqrt{2\gamma}}{W}\Big)^{-1/2}(1+o(1)).
 \hfill
$

Moreover, if we set
\begin{equation}\label{G}
G^{(m)}(\gamma)=\left(-\Delta+\dfrac{2\gamma}{W^2}\right)^{-1},
\end{equation}
then
\begin{equation}\label{G_as}
|G^{(m)}_{ii}(\gamma)|\le\dfrac{C_\gamma W}{\sqrt{2\gamma}}\coth\dfrac{m\sqrt{2\gamma}}{W}(1+o(1)).
\end{equation}


\item $ \dfrac{|Z^{(m)}_\gamma-Z_{\delta,\gamma}^{(m)}|}{|Z^{(m)}_\gamma|}=
|Z^{(m)}_\gamma|^{-1}\bigg|\,\displaystyle\int\limits_{\max |x_i|>\delta W}\mu^{(m)}_{\gamma}(x)
\prod\limits_{q=1}^md x_q\bigg|\le C_1\, e^{-C_2\delta^2 W},\quad W\to\infty, $ where $m> CW$,
$\delta=W^{-\kappa}$ for sufficiently small $\kappa<\theta/8$, and
\[
Z_{\delta,\gamma}^{(m)}= \intd_{-\delta W}^{\delta W} \mu_{\gamma}^{(m)}(x) \prod\limits_{q=1}^md x_q.
\]

In addition, for any $m$
\[
|Z^{(m)}_\gamma|^{-1}\bigg|\,\displaystyle\int\limits_{|x_k-x_1|>\delta W}\mu^{(m)}_{\gamma}(x)
\prod\limits_{q=1}^md x_q\bigg|\le C_1\, e^{-C_2\delta^2 W},\quad W\to\infty,
\]
and for $m>CW$ and any $\gamma_1,\gamma_2\in \mathbb{C}$, $\Re \gamma_1,\Re \gamma_2>0$
\begin{equation}\label{otn_mod}
\dfrac{|Z^{(m)}_{\gamma_1}|}{|Z^{(m)}_{\gamma_2}|}\le e^{C_1m/W}, \quad W\to\infty.
\end{equation}
\item Let $m>C_1W$, $k\le Cm/W$, $S=\{i_1,\ldots,i_s\} \subset \{1,\ldots,m\}$, and
$\sum\limits_{l=1}^sk_{i_l}=k$, where
$k_l\in \{1,\ldots,k\}$. Then $$
|Z^{(m)}_\gamma|^{-1}\bigg|\,\displaystyle\int\limits_{\max |x_i|>\delta W}\prod\limits_{j\in S}(x_j/W)^{k_j}\cdot \mu^{(m)}_{\gamma}(x)
\prod\limits_{q=1}^md x_q\bigg|\le e^{-C_1\delta^2 W},\quad W\to\infty,
$$
where $\delta=W^{-\kappa}$ for sufficiently small $\kappa<\theta/8$.
\end{enumerate}
\end{lemma}

Introduce the following partial ordering. Let $\Phi_1(x_1,\ldots,x_n)$, $\Phi_2(x_1,\ldots,x_n)$ be two analytic
functions in some ball centered at $0$, and let the coefficients of the Taylor expansion of $\Phi_2$ be
non-negative. Then we write
\begin{equation}\label{major}
\Phi_1\prec \Phi_2
\end{equation}
if the absolute value of each coefficient of the Taylor expansion of $\Phi_1$ does not exceed the corresponding
coefficient of $\Phi_2$.

 It is easy to see that
\begin{equation}\label{prop}
\Phi_3\prec \Phi_1,\quad \Phi_4\prec \Phi_2 \Rightarrow \Phi_3\Phi_4\prec \Phi_1\Phi_2.
\end{equation}

 We will need
\begin{lemma}[\cite{TSh:14}, Lemma 8]\label{l:maj}

(i) $\quad$ Let $\,|\phi_1|\le CW^{-1}$, $|\phi_2|=o(1)$ and $|\phi_k|\le C^k$ for some absolute
constant $C>0$. Then
\begin{align}\label{gen_phi}
&\langle\prod\limits_{i=-n}^n(1+\sum\limits_{l=1}^\infty |\phi_l|x_i^l/W^l)\rangle_{0,*}\le
\exp\{C|\phi_2|n/W\}.
\end{align}
In particular, for $\,|\phi_1|\le CW^{-1}$, $|\phi_2|=O(1/W)$ we have
\begin{equation*}
\langle\prod\limits_{i=-n}^n(1+\sum\limits_{l=1}^\infty |\phi_l|x_i^l/W^l)\rangle_{0,*}=1+o(1).
\end{equation*}

(ii) $\quad$ If
\begin{equation*}
\Phi_1(s_1,\ldots,s_n)-\Phi_1(0,\ldots,0)\prec \prod\limits_{j=1}^n(1+q(s_i))-1,
\end{equation*}
where $s_i=s(\tilde{a\vphantom{b}}_i/W,\tilde{a\vphantom{b}}_{i+1}/W,\ldots,\tilde{a\vphantom{b}}_{i+k}/W,
\tilde{b}_i/W,\tilde{b}_{i+1}/W,\ldots,\tilde{b}_{i+k}/W)$ is a polynomial with $s(0,\ldots,0)=0$, $k$ is
an $n$-independent constant, and $q(s)=\sum_{j=1}^{\infty}|c_j|s^j$ with $|c_1|\le CW^{-1}$, $|c_2|=o(1)$,
$|c_{\,l}|\le (C_0)^l$,
$l\ge 3$, then
\begin{equation*}
\big|\langle\Phi_1(s_1,\ldots,s_n)-\Phi_1(0,\ldots,0)\rangle_0\big|\le \langle\prod\limits_{j=1}^n(1+q(s_i^*))-1\rangle_{0,*}+e^{-Cn/W},
\end{equation*}
where $s_i^*$ is obtained from $s_i$ by replacing the coefficients of $s$ with their absolute values.
\end{lemma}

\subsection{Integration over the symplectic group $\mathring{Sp}(2)$}
\begin{proposition}\label{p:Its-Z}
(i) Let $C$ be a normal $2\times 2$ matrix with distinct eigenvalues $c_1$, $c_2$ and
$D=\hbox{diag}\{d_1,d_2\}$, $d_i\in \mathbb{C}$. Then
\begin{equation}\label{Its-Zub}
\int_{U(2)} \exp\{t\Tr CU^*DU\} d\,\mu(U)=
\dfrac{e^{t(c_1d_1+c_2d_2)}-e^{t(c_1d_2+c_2d_1)}}{t(c_1-c_2)(d_1-d_2)},
\end{equation}
where $t\in \mathbb{C}$ is some constant.

(ii) Let
\begin{equation}\label{F_sympl}
F=\left(\begin{array}{cc}
X&w_2\sigma\\
-\overline{w}_2\sigma&X^t
\end{array}\right), \quad
X=\left(\begin{array}{cc}
x&w_1\\
\overline{w}_1&y
\end{array}\right),
\end{equation}
\begin{equation*}
G=\left(\begin{array}{cc}
D&0\\
0&D
\end{array}\right),\quad
D=\left(\begin{array}{cc}
d_1&0\\
0&d_2
\end{array}\right),
\end{equation*}
where $\sigma$ is defined in (\ref{s-s'}), $x,y\in \mathbb{R}$, $w_1, w_2, d_1, d_2\in \mathbb{C}$.
The matrices of the form (\ref{F_sympl}) can be diagonalized by $\mathring{Sp}(2)$ transformation $P$ and
have two real eigenvalues $a, b$ of multiplicity two. Moreover, the measure
\[
dF=dx\,dy\,d\Re w_1\,d\Im w_1\,d\Re w_2\,d\Im w_2,
\]
can be represented in the form
\begin{equation*}
\dfrac{\pi^2}{12}(a-b)^4d\nu(P)
\end{equation*}
with
\begin{equation}\label{nu}
d\nu(P)=3(1-2|V_{12}|^2)^2d\mu(U)d\mu(V).
\end{equation}
Here $d\mu$ is a Haar measure over $\mathring{U}(2)$,
\begin{equation*}
P=\left(\begin{array}{cc}
V&0\\
0&\overline{V}
\end{array}\right)\cdot
\left(\begin{array}{cc}
\cos\varphi\cdot I&\sin\varphi\cdot e^{i\alpha}\cdot \sigma'\\
-\sin\varphi\cdot e^{-i\alpha}\cdot \sigma'&\cos\varphi\cdot I
\end{array}\right)
\end{equation*}
and
\begin{equation*}
U=\left(\begin{array}{cc}
\cos\varphi&\sin\varphi\cdot e^{i\alpha}\\
-\sin\varphi\cdot e^{-i\alpha}&\cos\varphi
\end{array}\right), \quad
V=\left(\begin{array}{cc}
\cos\phi&\sin\phi\cdot e^{i\beta}\\
-\sin\phi\cdot e^{-i\beta}&\cos\phi
\end{array}\right).
\end{equation*}

Moreover, if $\tilde{t}=t(c_1-c_2)(d_1-d_2)$, then
\begin{multline}\label{Its-Zub_s1}
\int_{\mathring{Sp}(2)} \exp\{t\Tr GP^*HP/2\} d\,\nu(P)\\
=\dfrac{6}{\tilde{t}^2}\,
\Big(e^{t(c_1d_1+c_2d_2)}\big(1-2/\tilde{t}\big)+
e^{t(c_1d_2+c_2d_1)}\big(1+2/\tilde{t}\big)\Big),
\end{multline}
In addition,
\begin{align}\label{Its-Zub_s}
&\int\limits_\Omega \exp\Big\{-\dfrac{t}{4}\,\Tr (F-G)^2\Big\}\, \Phi(F) dF\\ \notag
&=\dfrac{\pi^2}{t^2}\int\limits_{\hat{\Omega}} \exp\Big\{-\frac{t}{2}\,\Tr(\hat{Y}-D)^2\Big\}
\cdot\Phi(\hat{Y})\cdot \dfrac{(y_1-y_2)^2}
{(d_1-d_2)^2}\\ \notag
&\times \Big(1-\dfrac{2}{t(y_1-y_2)(d_1-d_2)}\Big)\,dy_1\,dy_2,
\end{align}
where $y_1$, $y_2$ are eigenvalues of $F$, $\hat{Y}=\hbox{diag}\,\{y_1,y_2\}$, and
\[
dF=dx\,dy\,d\Re w_1\,d\Im w_1\,d\Re w_2\,d\Im w_2.
\]
Here $\Phi(F)$ is any function which is invariant over $\mathring{Sp}(2)$ transformation (i.e., depend only on
$y_1$, $y_2$), $\Omega$
is any $\mathring{Sp}(2)$ invariant domain such that the eigenvalues of $F$ of the form (\ref{F_sympl}) run
over the symmetric
domain $\hat{\Omega}$. 
\end{proposition}
The proof of this proposition can be found in Sec.6.

\section{Proof of the main theorem}
In this section we will prove Theorem \ref{thm:1} applying the steepest descent method
to the integral representation (\ref{F_0_1}).

\subsection{The bound for $\Sigma_c$}
\begin{lemma}\label{l:2} Let $\Sigma_c$ be the part of the integral in (\ref{F_0_1}) over the complement
of the domain $\Omega_\delta$, which is defined in (\ref{Om_delta}). Then
\begin{equation*}
|\Sigma_c|\le C_1W^{-6N+4}(24\pi)^{N}e^{-2Nc_0} e^{-C_2W^{1-2\kappa}},
\end{equation*}
where $\kappa<\theta/8$ and $c_0=\Re f(a_\pm)$.
\end{lemma}
\begin{proof}
According to (\ref{F_0_1}), we have
\begin{align*}
|\Sigma_c|&\le e^{-2Nc_0}\cdot \int\limits_{\Omega_\delta^C}\exp\Big\{-
\sum\limits_{j=-n}^n(f_*(a_j)+f_*(b_j))\Big\}\\
&\times\exp\Big\{-\frac{W^2}{4}\sum\limits_{j=-n+1}^n\Tr
(Q_j^*A_{j,4}Q_j-A_{j-1,4})^2\Big\}\\
&\times\prod\limits_{l=-n}^n(a_l-b_l)^4d\,\nu(P_{-n})\,d\overline{a\vphantom{b}}\,
d\overline{b}\, \prod\limits_{p=-n+1}^nd\nu(R_p),
\end{align*}
where $f_*$ and $c_0$ are defined in (\ref{f}) and (\ref{c_pm}). Here we insert the absolute value inside the integral
and use that
\[
\Big|\exp\Big\{-
\frac{i}{2N\rho(\lambda_0)}\sum\limits_{j=-n}^n\Tr \big(R_jP_{-n}\big)^*A_{j,4}\,
(R_jP_{-n}\big)\hat{\xi}_4\Big\}\Big|=1.
\]
To simplify formulas below, set
\begin{equation}\label{I_0}
I_0= W^{-6N+4}(24\pi)^{N}e^{-2Nc_0}\cdot \left|\mdet^{-1}\left(-\Delta+2c_+/W^2\right)\right|.
\end{equation}
As we will see below, $I_0$ is an order of $\Sigma$ (see Lemma \ref{l:sigma}).
Also recall that, according to Lemma \ref{l:okr} (1),
\begin{equation}\label{in_det}
e^{-C_1N/W}\le \left|\mdet^{-1}\left(-\Delta+2c_+/W^2\right)\right|\le e^{-C_2N/W},
\end{equation}
and that $W^2=N^{1+\theta}$, $\kappa<\theta/8$,
and hence $CN/W\ll W^{1-2\kappa}$.

We are going to prove that
\begin{equation}\label{in_i_0}
|\Sigma_c/I_0|\le e^{-CW^{1-2\kappa}}.
\end{equation}
Using (\ref{Its-Zub_s}), we get (recall that
$A_j=\hbox{diag}\,\{a_j,b_j,a_j,b_j\}$, $j=-n,\ldots,n$ and $\Omega_\delta^C$ is still a symmetric domain)
\begin{align} \notag
&I_0^{-1}\cdot|\Sigma_c|\le \dfrac{12^{N-1}e^{-2Nc_0}}{W^{4(N-1
)}I_0}\int\limits_{\Omega_\delta^C}
\exp\Big\{-\frac{W^2}{2}\sum\limits_{j=-n+1}^n\Big(
(a_j-a_{j-1})^2+
(b_j-b_{j-1})^2\Big)\Big\}\\ \label{sig_vn1}
&\times\exp\Big\{-\sum\limits_{j=-n}^n(f_*(a_j)+f_*(b_j))\Big\}\,
 (a_{-n}-b_{-n})^2(a_n-b_n)^2\\ \notag
&\times \prod\limits_{j=-n}^n\Big(1-\dfrac{2}{W^2(a_j-b_j)(a_{j-1}-b_{j-1})}\Big)
\,d\overline{a\vphantom{b}}\, d\overline{b}\\ \notag
\end{align}
The first line here
is obtained performing
recursively the integral over $Q_j$ starting from $j=n$ and going backwards. At each step the
integral can be written in the form (\ref{Its-Zub}), with a suitable choice of the
function $f$. The last product of (\ref{sig_vn1}) can be bounded by $\exp\{CN/W^2\}$, thus
\begin{align} \notag
&I_0^{-1}\cdot|\Sigma_c|\le\dfrac{12^{N-1}e^{-2Nc_0} \cdot e^{CN/W^2}}{W^{4(N-1)}I_0}\int\limits_{\Omega_\delta^C}
\exp\Big\{-\frac{W^2}{2}\sum\limits_{j=-n+1}^n\Big(
(a_j-a_{j-1})^2+
(b_j-b_{j-1})^2\Big)\Big\}\\ \label{sig_vn}
&\times\exp\Big\{-\sum\limits_{j=-n}^n(f_*(a_j)+f_*(b_j))\Big\}\,
 (a_{-n}-b_{-n})^2(a_n-b_n)^2\,d\overline{a\vphantom{b}}\, d\overline{b}\\ \notag
&\le C\cdot W^4 \cdot (2\pi)^{-N} e^{C_1N/W}\int\limits_{W\Omega_\delta^C}\exp\Big\{-\frac{1}{2}
 \sum\limits_{j=-n+1}^n\Big(
(a_j-a_{j-1})^2+(b_j-b_{j-1})^2\Big)\Big\}\\ \notag
&\times\exp\Big\{-\sum\limits_{j=-n}^n(f_*(a_j/W)+f_*(b_j/W))\Big\}\,
 (a_{-n}-b_{-n})^2(a_n-b_n)^2\,d\overline{a\vphantom{b}}\, d\overline{b},
\end{align}
where $f_*$ and $c_0$ are defined in (\ref{f}) and (\ref{c_pm}). Here in the third line we did the change
$a_j\to a_j/W$, $b_j\to b_j/W$ and used (\ref{I_0}) -- (\ref{in_det}).

Now the last integral in (\ref{sig_vn}) is the same as in \cite{TSh:14}, eq. (5.5) and so
can be bounded in the same way.

\end{proof}

\subsection{Calculation of $\Sigma$}
\begin{lemma}\label{l:sigma}
For the integral $\Sigma$ over the domain $\Omega_\delta$ (see (\ref{Om_delta})) we have
\begin{align}\label{sigma}
\Sigma&=\dfrac{8 \pi^4\rho(\lambda_0)^4 e^{-2Nc_0}(24\pi)^{N}}{3W^{6N-4}}\cdot
DS(\pi(\xi_1-\xi_2))
\cdot\Big|\mdet^{-1}\Big(-\Delta+\frac{2c_+}{W^2}\Big)\Big|(1+o(1))\\ \notag
&=8(\pi\rho(\lambda_0))^4/3\cdot DS(\pi(\xi_1-\xi_2))\cdot I_0 ,\quad W\to\infty,
\end{align}
where $I_0$ is defined in (\ref{I_0}).
\end{lemma}
Note that (\ref{sigma}) together with (\ref{in_i_0}) yield
\[
|\Sigma_c|\le e^{-CW^{1-2\kappa}} |\Sigma|,
\]
which gives (\ref{F_2}).

Now using (\ref{F_2}) and (\ref{sigma}) we get Theorem \ref{thm:1}.

Thus, we are left to compute $\Sigma$. We are going to show that the leading term in $\Sigma$
is given by $\Sigma_\pm$, i.e., that the contributions of $\Sigma_+$ and $\Sigma_-$ are smaller.

\subsubsection{Calculation of $\Sigma_\pm$}
Consider the $\delta$-neighborhood of the point $(\overline{a}_+,\overline{a}_-)$ with $\overline{a}_{\pm}$ of
(\ref{a_pm}) and $\delta=W^{-\kappa}$.

Let us show that
\begin{lemma}\label{l:sig}
For the integral $\Sigma_\pm$ over the domain $\Omega_\delta^\pm$ of (\ref{Om_delta}) we have, as $W\to\infty$
\begin{align*}\notag
\Sigma_{\pm}=\dfrac{8 (\pi\rho(\lambda_0))^4 e^{-2Nc_0}(24\pi)^{N}}{3W^{6N-4}}\cdot
DS(\pi(\xi_1-\xi_2))\cdot \Big|\mdet^{-1}\Big(-\Delta+\frac{2c_+}{W^2}\Big)\Big|(1+o(1)).
\end{align*}
\end{lemma}
\begin{proof}
%
%

Performing the change $a_j-a_+=\tilde{a\vphantom{b}}_j/W$, $b_j-a_-=\tilde{b}_j/W$ in (\ref{F_0_1})
and using (\ref{f_exp}),
we obtain (recall that $a_\pm=\pm\pi\rho(\lambda_0)$)
\begin{align}\notag
\Sigma_{\pm}=&\dfrac{2^{N}e^{-2Nc_0-i\pi(\xi_1-\xi_2)}}{W^{2N}}
\int\limits_{|\tilde{a\vphantom{b}}_j|,|\tilde{b}_j|\le W^{1-\kappa}}
\mu_{c_+}(a)\mu_{c_-}(b)\cdot e^{-\sum\limits_{k=-n}^n\big(\varphi_+(\tilde{a}_k/W)
+\varphi_-(\tilde{b}_k/W)\big)}
\\ \label{F_0_2}
&\times \int\limits_{\mathring{Sp}(2)^N} e^{W^2/2\sum\limits_{j=-n+1}^n\Tr \left(Q_j^*(L_4+\tilde{A}_{j,4}/W)Q_{j}(L_4+\tilde{A}_{j-1,4}/W)-
(L_4+\tilde{A}_{j,4}/W)(L_4+\tilde{A}_{j-1,4}/W)\right)}\\ \notag
& \times
\exp\Big\{-\frac{i}{2N\rho(\lambda_0)}\sum\limits_{k=-n}^n\big(\Tr (R_kP_{-n})^*(L_4+\tilde{A}_{k,4}/W)\,
(R_kP_{-n})\hat{\xi}_4-\Tr L_4\hat{\xi}_4\big)\Big\}
\\ \notag
&\times\prod\limits_{l=-n}^n(a_+-a_-+(\tilde{a\vphantom{b}}_l-\tilde{b}_l)/W)^4d\nu(P_{-n})
\prod\limits_{q=-n+1}^n
d\nu(Q_q)\,d\overline{a\vphantom{b}}
\,d\overline{b},
\end{align}
where $L_4=\hbox{diag}\,\{a_+,a_-,a_+,a_-\}$, $\tilde{A}_{j,4}=\hbox{diag}\,
\{\tilde{a\vphantom{b}}_j,\tilde{b}_j, \tilde{a\vphantom{b}}_j,\tilde{b}_j\}$, and $\mu_{\gamma}(a)$
is defined in (\ref{mu}).

Now we are going to integrate over $\{Q_j\}$. Introduce
\begin{align}\label{F_V}
&F(\overline{\vphantom{b} a},\overline{b},Q)=-\frac{i}{2\rho(\lambda_0)}\sum\limits_{k=-n}^n\big(\Tr
\big(R_kP_{-n}\big)^*(L_4+\tilde{A}_{k,4}/W)\,
(R_kP_{-n}\big)\hat{\xi}_4-\Tr L_4\hat{\xi}_4\big),\\ \notag
&d\,{\vphantom{A}\eta}_j(Q_j,\tilde{A}_j)=e^{\frac{W^2}{2}\Tr \left(Q_j^*
(L_4+\tilde{A}_{j,4}/W)Q_{j}(L_4+\tilde{A}_{j-1,4}/W)-
(L_4+\tilde{A}_{j,4}/W)(L_4+\tilde{A}_{j-1,4}/W)\right)}d\nu(Q_j),\\ \notag
&d\,{\vphantom{A}\eta}(Q,\tilde{A})=\prod\limits_{j=-n+1}^nd\,{\vphantom{A}\eta}_j(Q_j,\tilde{A}_j),\quad
I_{\eta}(\tilde{A})=\int d\,\eta(Q,\tilde{A}),\\ \notag
&t_j=W^2\Big(a_+-a_-+(\tilde{a\vphantom{b}}_j-\tilde{b}_j)/W\Big)\Big(a_+-a_-+
(\tilde{a\vphantom{b}}_{j-1}-\tilde{b}_{j-1})/W\Big),\quad q_j=6/t_j^2.
\end{align}
According to Proposition \ref{p:Its-Z} we have
\begin{equation}\label{int_V}
I_{\eta}(\tilde{A})=\prod\limits_{j=-n+1}^n q_j
\left[1-\dfrac{2}{t_j}+e^{-t_j}
\Big(1+\dfrac{2}{t_j}\Big)\right].
\end{equation}
We want to integrate the r.h.s. of (\ref{F_0_2}) over $d\eta (Q,\tilde{A})$.
To this end, we expand \\$\exp\big\{F(\overline{\vphantom{b} a},\overline{b},Q)\big\}$ into
a series with respect to the elements of $Q_j$, $j=-n+1,\ldots,n$. We are going to show that the leading term of the integral is given by the summands without
any elements of $Q_j$.
\begin{lemma}\label{l:un} In the notations of (\ref{F_V})
\begin{equation}\label{F_ch}
\Big|
\Big\langle \Big\langle\big(\exp\{(F(\overline{\vphantom{b} a},\overline{b},Q)-F(0,0,I))/N\}-1\big)
\cdot\displaystyle\Pi_1\cdot\Pi_2\Big\rangle_{0}\Big\rangle_{\eta}\Big|=o(1),\quad N\to\infty,
\end{equation}
where $\Pi_1$, $\Pi_2$ are the products of the Taylor's series for $\exp\{\varphi_+(\tilde{a}_j/W)\}$ and
for \\ $\exp\{\varphi_-(\tilde{b}_j/W)\}$ and
\begin{equation}\label{ang_mu}
\langle\ldots\rangle_{\eta}
=\Big(\prod\limits_{j=-n+1}^nq_j\Big)^{-1}\intd (\ldots) d\eta(Q,\tilde{A}).
\end{equation}
\end{lemma}
\begin{proof}
Since $\hat{\xi}_4=\frac{\xi_1+\xi_2}{2}\, I_4+\frac{\xi_1-\xi_2}{2}\, L_4$ and
$a_+=-a_-=\pi\rho(\lambda_0)$, we have
\begin{align*}
&\Tr (R_kP_{-n})^*(a_+L_4+\tilde{A}_{k,4}/W)
(R_kP_{-n})\hat{\xi}_4-\Tr (a_+L_4+\tilde{A}_{k,4}/W)\hat{\xi}_4\\
&=\frac{\xi_1-\xi_2}{2}\,\Tr ((R_kP_{-n})^*
(a_+L_4+\tilde{A}_{k,4}/W)\,
(R_kP_{-n})L_4-(a_+L_4+\tilde{A}_{k,4}/W)L_4)\\
&=4\pi\rho(\lambda_0)(\xi_2-\xi_1)(1+
(\tilde{a\vphantom{b}}_k-\tilde{b}_k)/2\pi\rho(\lambda_0)W)\cdot
(|(R_kP_{-n})_{12}|^2+|(R_kP_{-n})_{14}|^2).
\end{align*}
For any $4\times 4$ matrix $P$ introduce
\begin{equation}\label{S_q}
S(P)=|P_{12}|^2+|P_{14}|^2.
\end{equation}
Note that for $P\in Sp(2)$ we have $S(P)\in [0,1]$.

Rewrite
\begin{multline}\label{razn}
F(\overline{\vphantom{b} a},\overline{b},Q)-F(0,0,I)\\
=2i\pi(\xi_1-\xi_2)
\sum\limits_{k=-n+1}^n\left(S(R_kP_{-n})-S(P_{-n})\right)\cdot \Big(1+
\dfrac{\tilde{a\vphantom{b}}_k-\tilde{b}_k}{2\pi\rho(\lambda_0)W}\Big).
\end{multline}
Thus, we get
\begin{multline*}
\Big\langle\exp\Big\{\dfrac{1}{N}\Big(F(\overline{\vphantom{b} a},\overline{b},Q)-F(0,0,I)\Big)\Big\}-1\Big
\rangle_\eta\\
=\sum\limits_{p=1}^\infty
\dfrac{C^p}{p!\,N^p}\sum\limits_{k_1,\ldots,k_p} \Big\langle
\prod\limits_{j=1}^p \Big[\Big(S(R_{k_j}P_{-n})-S(P_{-n})\Big)\cdot \Big(1+
\dfrac{\tilde{a\vphantom{b}}_{k_j}-\tilde{b}_{k_j}}{2\pi\rho(\lambda_0)W}\Big)\Big]\Big\rangle_{\eta},
\end{multline*}
where $\langle\ldots\rangle_{\eta}$ is defined in (\ref{ang_mu}).
Hence, we have to study
\begin{equation}\label{Phi}
\Phi_{k_1,\ldots,k_p}(\overline{\vphantom{b} a},\overline{b})=\Big\langle\prod\limits_{j=1}^p
\Big(S(R_{k_j}P_{-n})-S(P_{-n})\Big)\Big\rangle_{\eta}.
\end{equation}
Let $p<Cn/W$ for some constant $C$. Introduce i.i.d vectors $\{(x_j, y_j)\}$ such that
the density of the distribution has the form
\begin{equation}\label{rho_t}
\rho(x_j,y_j)=4(a_+-a_-)^4\, x_jy_j\, \exp\{-(a_+-a_-)^2[x_j^2+y_j^2]\}\cdot
\mathbf{1}_{0<x_j,y_j<W/2}.
\end{equation}
Introduce matrices
\begin{equation*}
\widetilde{Q}_j=\mathcal{V}_j\cdot \mathcal{U}_j,
\end{equation*}
where
\begin{align}\label{UV}
\mathcal{V}_j&=\left(\begin{array}{cc}
\widetilde{V}_j& 0\\
0 &\overline{\widetilde{V}}_j
\end{array}\right), \quad \widetilde{V}_j=\left(\begin{array}{cc}
\tilde{r}_je^{i\tilde{\sigma}_j}& \tilde{v}_je^{i\sigma_j}\\
-\tilde{v}_je^{-i\sigma_j}&\tilde{r}_je^{-i\tilde{\sigma}_j}
\end{array}\right),\\ \notag
\mathcal{U}_j&=\left(\begin{array}{cc}
\tilde{t}_j e^{i\tilde{\theta}_j}I& \tilde{u}_je^{i\theta_j}\sigma^\prime\\
-\tilde{u}_je^{-i\theta_j}\sigma^\prime &\tilde{t}_j e^{-i\tilde{\theta}_j}I
\end{array}\right)
\end{align}
with
\begin{align*}
\tilde{v}_j&=x_j/p_jW,\quad \tilde{u}_j=y_j(1-2\tilde{v}_j^2)^{-1/2}/p_jW\\
p_j&=\Big(1+\dfrac{\tilde{a\vphantom{b}}_j-\tilde{b}_j}{W(a_+-a_-)}\Big)^{1/2}
\Big(1+\dfrac{\tilde{a\vphantom{b}}_{j-1}-\tilde{b}_{j-1}}{W(a_+-a_-)}\Big)^{1/2},\\
\tilde{r}_j&=(1-\tilde{v}_j^2)^{1/2},\quad \tilde{t}_j=(1-\tilde{u}_j^2)^{1/2}
\end{align*}
and $\theta_j,\tilde{\theta}_j,\sigma_j,\tilde{\sigma}_j\in [-\pi,\pi)$. Define also
\begin{align}\label{eta_til}
d\tilde{\eta}_j=(2\pi)^{-4}\rho(x_j,y_j)dx_j\,d y_j\,
d\theta_j\,d\tilde{\theta}_j\,d\sigma_j\,d\tilde{\sigma}_j,\quad
d\tilde{\eta}=\prod\limits_{j=-n+1}^n d\tilde{\eta}_j.
\end{align}
Note that 
\[
\int d\tilde{\eta}_j=\big(1-e^{-W^2(a_+-a_-)^2/4}\big)^2\le 1.
\]
We need
\begin{lemma}\label{l:phi}
\begin{multline}\label{Phi_tild}
\widetilde{\Phi}_{k_1,\ldots,k_p}(\overline{\vphantom{b}a},\overline{b}):=\Big\langle\prod\limits_{j=1}^p
\Big( S\Big(\tilde{R}_{k_j}\cdot P_{-n}\Big)-S(P_{-n})\Big)\cdot
\prod\limits_{i=-n+1}^n(1-2|(\tilde{V}_i)_{12}|^2)
\Big\rangle_{\tilde{\eta}}\\=
\Phi_{k_1,\ldots,k_p}(\overline{\vphantom{b} a},\overline{b})+O(e^{-cW^2}),
\end{multline}
where $\langle\ldots\rangle_{\tilde{\eta}}$ means the integration over $d\tilde{\eta}$ and
\[
\widetilde{R}_{k_j}=\prod\limits^{-n+1}_{l=k_j}\widetilde{Q}_l.
\]
\end{lemma}
The proof of the lemma can be found in Sec. 6.

Denote
\begin{equation}\label{s_j}
s_j=1-\Big(1+\dfrac{\tilde{a\vphantom{b}}_j-\tilde{b}_j}{W(a_+-a_-)}\Big)
\Big(1+\dfrac{\tilde{a\vphantom{b}}_{j-1}-\tilde{b}_{j-1}}{W(a_+-a_-)}\Big).
\end{equation}
Expanding $\mathcal{V}_j$, $\mathcal{U}_j$ of (\ref{UV}) with respect to $s_j$ we get
\[
\mathcal{V}_j=
\left(\begin{array}{cc}
\widetilde{V}_j(0)&0\\
0&\overline{\widetilde{V}}_j(0)
\end{array}\right)
+\dfrac{x_j}{W}\cdot g_v(s_j)\cdot
\left(\begin{array}{cc}
V_j^{1}&0\\
0&\overline{V}_j^{1}
\end{array}\right)
+\dfrac{x_j^2}{W^2}\sum\limits_{r=1}^\infty
\left(\begin{array}{cc}
V^{(r)}_j&0\\
0&\overline{V}^{(r)}_j
\end{array}\right)
s_j^r,
\]
\begin{align*}
\mathcal{U}_j= \mathcal{U}_j(0)+
\dfrac{y_j}{W}\cdot h_u(s_j)\cdot
\left(\begin{array}{cc}
0&e^{i\theta_j}\sigma'\\
-e^{-i\theta_j}\sigma'&0
\end{array}\right)
+\dfrac{y_j^2}{W^2}\sum\limits_{r=1}^\infty
\left(\begin{array}{cc}
U^{(r)}_j&0\\
0&\overline{U}^{(r)}_j
\end{array}\right)
s_j^r,
\end{align*}
where
\begin{equation}
g_v(s_j)=(1-s_j)^{-1/2}-1,\quad h_u(s_j)=
(1-2x_j^2/W^2)^{-1/2}\Big(\Big(1-\dfrac{s_j}{1-2x_j^2/W^2}\Big)^{-1/2}-1\Big).
\end{equation}
Here $\widetilde{V}_j(0)$, $\mathcal{U}_j(0)$ are unitary matrices (and hence $\|\widetilde{V}_j(0)\|\le 1$,
$\|\mathcal{U}_j(0)\|\le 1$),
\[
\widetilde{V}_j^{1}=\left(
\begin{array}{cc}
0&e^{i\sigma_j}\\
-e^{-i\sigma_j}&0
\end{array}\right),\quad \|\widetilde{V}^{(r)}_j\|\le C^r,\quad \|\widetilde{U}^{(r)}_j\|\le C^r \quad (r=1,2,\ldots),
\]
and $\{\widetilde{V}^{(r)}_j\}$, $\{\widetilde{U}^{(r)}_j\}$ are diagonal matrices.

Since the integrals of $e^{im\theta_j}$ equal 0 for $m\ne 0$ and $2\pi$ for $m=0$, we conclude that if we
replace the coefficients in front of $e^{i\theta_j}$ and $e^{-i\theta_j}$ with the bounds for their absolute values,
then, after the averaging with respect to $\theta_j$, the resulting coefficients in front of $s_j^k$ will
grow. The same is true for the integral with respect to $\sigma_j$. Moreover, for $x_j\in (0,W/2)$
\begin{align*}
g_v(s_j)&\prec g_v^1(s_j^*):=\dfrac{C_1}{1-C_2s_j^*},\\
h_u(s_j)&\prec h_u^1(s_j^*):=\dfrac{C_3}{1-C_4s_j^*}
\end{align*}
where $C_l$, $l=1,\ldots, 4$ are $n$-independent constant and
\begin{equation*}
s_j^*=\dfrac{\tilde{a\vphantom{b}}_j+\tilde{b}_j+
\tilde{a\vphantom{b}}_{j-1}+\tilde{b}_{j-1}}{W(a_+-a_-)}
+\dfrac{(\tilde{a\vphantom{b}}_{j-1}+\tilde{b}_{j-1})
(\tilde{a\vphantom{b}}_j+\tilde{b}_j)}{W^2(a_+-a_-)^2}.
\end{equation*}
Hence,
\begin{multline*}
\widetilde{\Phi}_{k_1,\ldots,k_p}(\overline{\vphantom{b}a},\overline{b})-
\widetilde{\Phi}_{k_1,\ldots,k_p}(0,0)\\
\prec \Big\langle\Big(
\hbox{Prod}_p(x,\sigma)\hbox{Prod}_p(y,\theta)-1\Big)\prod\limits_{j=-n}^n\Big(1-\dfrac{2x_j^2}{W^2}+\dfrac{x_j^2}{W^2}
\cdot\dfrac{s_j^*}{1-s_j^*}\Big)\Big\rangle_{x_j,y_j,\sigma_j,\theta_j},
\end{multline*}
where
\begin{align*}
\hbox{Prod}_p(x,\sigma)=\prod \Big|1+\frac{x_j}{W}\, e^{i\sigma_j} s^*_jg(s^*_j)+\frac{x_j^2}{W^2}s^*_j
g(s^*_j)\Big|^{2p},\\
\hbox{Prod}_p(y,\theta)=\prod \Big|1+\frac{y_j}{W}\, e^{i\theta_j} s^*_jh(s^*_j)+\frac{y_j^2}{W^2}s^*_j
h(s^*_j)\Big|^{2p}
\end{align*}
and $g(t)$ and $h(t)$ are the function of the form $C_1/(1-C_2t)$ with positive $n$-independent
$C_1$, $C_2$ (we denote the set of such function by $\mathcal{L}[t]$).

In addition,
\[
\Big\langle
\dfrac{x_j^{2k}}{W^{2k}}\Big\rangle_{x_j}\le \dfrac{k!}{(a_+-a_-)^{2k}W^{2k}},\quad
\Big\langle
\dfrac{y_j^{2k}}{W^{2k}}\Big\rangle_{y_j}\le \dfrac{k!}{(a_+-a_-)^{2k}W^{2k}},
\]
and thus we conclude
\begin{align*}
\Big\langle\hbox{Prod}_p(x,\sigma)\cdot\prod\limits_{j=-n}^n\Big(1-\dfrac{2x_j^2}{W^2}+\dfrac{x_j^2}{W^2}
\cdot\dfrac{s_j^*}{1-s_j^*}&\Big)\Big\rangle_{x_j,\sigma_j}\\
&\prec
\prod \Big(1+\frac{p}{W^2}s_j^*g_1(s_j^*)+
\frac{p^2}{W^2}(s_j^*)^2
g_1(s_j^*)^2\Big),\\
\Big\langle\hbox{Prod}_p(y,\theta)\Big\rangle_{y_j,\theta_j}
\prec \prod \Big(1+\frac{p}{W^2}s_j^*h_1(s_j^*)+&
\frac{p^2}{W^2}(s_j^*)^2
h_1(s_j^*)^2\Big),
\end{align*}
where $g_1, h_1\in \mathcal{L}[t]$. Hence, we obtain 
\begin{align}\label{maj}
\widetilde{\Phi}_{k_1,\ldots,k_p}(\overline{\vphantom{b}a},\overline{b})-
\widetilde{\Phi}_{k_1,\ldots,k_p}(0,0)
\prec& \Big[\prod \Big(1+\frac{p}{W^2}s_j^*f_1(s_j^*)+
\frac{p^2}{W^2}(s_j^*)^2
f_1(s_j^*)^2\Big)-1\Big]\\ \notag
&+\Big[\prod \Big(1+\dfrac{2}{W^2}+\frac{1}{W^2}\dfrac{s_j^*}{1-s_j^*}\Big)-1\Big]
\end{align}
with some $f_1\in \mathcal{L}[t]$.

Set
\[
\Pi_3=\prod\limits_{j=1}^p \Big(1+
\dfrac{\tilde{a\vphantom{b}}_{k_j}-\tilde{b}_{k_j}}{(a_+-a_-)W}\Big),\quad
\Pi_{3,*}=\prod\limits_{j=1}^p\Big(1+\dfrac{\tilde{a\vphantom{b}}_{k_j}+\tilde{b}_{k_j}}{(a_+-a_-)W}\Big).
\]
Then Lemma \ref{l:maj} and (\ref{maj}) yield
\begin{align}\label{bphi}
&\Big|\Big\langle(\widetilde{\Phi}_{k_1,\ldots,k_p}(\overline{\vphantom{b}a},\overline{b})-
\widetilde{\Phi}_{k_1,\ldots,k_p}(0,0))\cdot
\displaystyle\Pi_1\cdot\Pi_2\cdot \Pi_3\Big\rangle_0\Big|\\ \notag
&\le \Big\langle\Big(\prod \Big(1+\frac{2p}{W^2}s_jf_1(s_j)+\frac{p^2}{W^2}s_j^2
f_1(s_j)^2\Big)-1\Big)
\cdot\displaystyle\Pi_{1,*}\cdot\Pi_{2,*}\cdot \Pi_{3,*}\Big\rangle_{0,*}\\ \notag
&+\Big\langle\Big(\prod \Big(1+\dfrac{2}{W^2}+\frac{1}{W^2}\dfrac{s_j^*}{1-s_j^*}\Big)-1\Big)
\cdot\displaystyle\Pi_{1,*}\cdot\Pi_{2,*}\cdot \Pi_{3,*}\Big\rangle_{0,*}+e^{-Cn/W}
\end{align}
Since $p\le Cn/W$, we have $2p/W^2\le W^{-1}$, $p^2/W^2=o(1)$. In addition, $\Pi_3$ has degree
$p<Cn/W$, $|\Pi_3|\le (1+\delta)^p$. Hence, we can write
\begin{align*}
&\Big\langle\Big(\prod \Big(1+\frac{2p}{W^2}s_jf_1(s_j)+\frac{p^2}{W^2}s_j^2
f_1(s_j)^2\Big)-1\Big)
\cdot\displaystyle\Pi_{1,*}\cdot\Pi_{2,*}\cdot \Pi_{3,*}\Big\rangle_{0,*}\\
&\le (1+\delta)^p\Big\langle\Big(\exp\Big\{\sum\limits_{i=-n}^n\Big(\frac{Cp}{W^2}
\cdot\frac{\tilde{a}_i+\tilde{b}_i}{W}+
\frac{p^2c}{W^2}\cdot\frac{\tilde{a}_i^2+\tilde{b}_i^2}{W^2}\Big)\Big\} -1\Big)
\cdot\displaystyle\Pi_{1,*}\cdot\Pi_{2,*}\Big\rangle_{0,*}\\
&\le e^{\delta p}\Big\langle\Big(\exp\Big\{\sum\limits_{i=-n}^n\Big(\frac{Cp}{W^2}\cdot
\frac{\tilde{a}_i+\tilde{b}_i}{W}+
\frac{p^2c}{W^2}\cdot\frac{\tilde{a}_i^2+\tilde{b}_i^2}{W^2}\Big)\Big\} -1\Big)^2
\Big\rangle_{0,*}^{1/2}
\cdot\Big\langle\displaystyle\Pi_{1,*}^2\cdot\Pi_{2,*}^2\Big\rangle_{0,*}^{1/2},
\end{align*}
where $\Pi_1$, $\Pi_2$ are the products of the Taylor's series for $\exp\{\varphi_+(\tilde{a}_j/W)\}$ and
for \\ $\exp\{\varphi_-(\tilde{b}_j/W)\}$, and $\Pi_{1,*}$, $\Pi_{2,*}$ are obtained form $\Pi_1$, $\Pi_2$ by
changing the coefficients to their absolute values.

 The second factor is $1+o(1)$ (see Lemma \ref{l:maj}(i)). Moreover, taking
the Gaussian integral of the first
factor, we obtain
\begin{multline*}
\Big\langle\Big(\prod \Big(1+\frac{2p}{W^2}s_jf_1(s_j)+\frac{p^2}{W^2}s_j^2
f_1(s_j)^2\Big)-1\Big)
\cdot\displaystyle\Pi_{1,*}\cdot\Pi_{2,*}\cdot \Pi_{3,*}\Big\rangle_{0,*}\\
\le e^{\delta p}
\Big(\exp\Big\{\frac{cp^2n}{W^3}\Big\}-1\Big)\le e^{\delta p}\Big(\exp\Big\{\frac{cpn^2}{W^4}\Big\}-1\Big).
\end{multline*}
Similarly,
\begin{multline*}
\Big\langle\Big(\prod \Big(1+\dfrac{2}{W^2}+\frac{1}{W^2}\dfrac{s_j^*}{1-s_j^*}\Big)-1\Big)
\cdot\displaystyle\Pi_{1,*}\cdot\Pi_{2,*}\cdot \Pi_{3,*}\Big\rangle_{0,*}
\le  e^{\delta p}\Big(\exp\Big\{\frac{cn}{W^4}\Big\}-1\Big).
\end{multline*}
Thus, since $p<Cn/W$, in view of (\ref{bphi}), we get
\begin{multline}\label{phi_bound}
\sum\limits_{p=1}^{Cn/W}\dfrac{(C_1)^p}{p!N^p}\sum\limits_{k_1,\ldots,k_p} \Big|
\Big\langle(\widetilde{\Phi}_{k_1,\ldots,k_p}
(\overline{\vphantom{b}a},\overline{b})-\widetilde{\Phi}_{k_1,\ldots,k_p}(0,0))
\cdot\displaystyle\Pi_1\cdot\Pi_2\cdot\Pi_3\Big\rangle_0\Big|\\
\le \exp\{C_1e^{\delta+C_2n^2/W^4}\}-e^{C_1e^\delta}+\Big(\exp\Big\{\frac{cn}{W^4}\Big\}-1\Big)
\cdot\Big(e^{C_1e^\delta}-1\Big)=o(1).
\end{multline}
If $p\gg n/W$, then $1/\sqrt{p!}\ll e^{-Cn/W}$, and hence we can replace $\langle\ldots\rangle_0$ with
$\langle\ldots\rangle_{0,*}$
(see Lemma~\ref{l:okr}) and then take the absolute value under the integral and get the bound
\[e^{C_1n/W}([\sqrt{Cn/W}]!)^{-1}\sum\limits_{p=CN/W}^\infty (C_2)^p/\sqrt{p!}=o(1).\]
Let us prove now that
\begin{equation}\label{phi0}
\widetilde{\Phi}_{k_1,\ldots,k_p}(0,0)=\Big\langle\prod\limits_{j=1}^p
\Big( S(\widetilde{R}_{k_j}(0)P_{-n})-S(P_{-n})\Big)\cdot \prod\limits_{i=-n+1}^n
(1-2|\tilde{V}_i(0)_{12}|^2)\Big\rangle_{\tilde{\eta}}=o(1).
\end{equation}
For the simplicity let us write
\[
\widetilde{R}_{k}^0:=\widetilde{R}_{k}(0), \quad \tilde{Q}_{k_1}^0:=\tilde{Q}_{k_1}(0),\quad \tilde{V}_i^0:=\tilde{V}_i(0).
\]
Note that $S(P)\in [0,1]$ for $P\in \mathring{Sp}(2)$, and thus
\begin{eqnarray}\notag
&\big|1-2S(P)\big|\le 1,\\  \label{b1}
&\big|S(P_1)-S(P_2)\big|\le 1,\quad P_1,P_2\in \mathring{Sp}(2),\\ \notag
&\big|1-2|\tilde{V}_i(0)_{12}|^2\big|\le 1.
\end{eqnarray}
Hence, we have
\begin{multline}\label{phi_0}
\Big|\widetilde{\Phi}_{k_1,\ldots,k_p}(0,0)\Big|\le \Big\langle
\Big| S(\widetilde{R}_{k_1}^0P_{-n})-S(P_{-n})\Big|
\Big\rangle_{\tilde{\eta}}
\le \Big\langle
\Big( S(\widetilde{R}_{k_1}^0P_{-n})-S(P_{-n})\Big)^2
\Big\rangle_{\tilde{\eta}}^{1/2}.
\end{multline}
In addition, 
\begin{align}\label{b2}
\Big\langle
&\Big( S(\widetilde{R}_{k_1}^0P_{-n})-S(P_{-n})\Big)^2
\Big\rangle_{\tilde{\eta}}
= \Big\langle
\Big(\big(S(\widetilde{R}_{k_1-1}^0P_{-n})-S(P_{-n})\big)\\ \notag
&+S(\tilde{Q}_{k_1}^0)\big(1-2S(\widetilde{R}_{k_1-1}^0P_{-n})\big)+H\big(\tilde{Q}_{k_1}^0,\widetilde{R}_{k_1-1}^0 \big)\Big)^2\Big\rangle_{\tilde{\eta}},
\end{align}
where
\begin{eqnarray}
H(P,Q)=\sum\limits_{l\ne s} P_{1l}Q_{l2}\overline{P}_{1s}\overline{Q}_{s2}+\sum\limits_{l\ne s} P_{1l}Q_{l4}\overline{P}_{1s}\overline{Q}_{s4}.
\end{eqnarray}
Since it is easy to check that
\begin{eqnarray*}
&\Big\langle
\Big(\big(S(\widetilde{R}_{k_1-1}^0P_{-n})-S(P_{-n})\big)
+S(\tilde{Q}_{k_1}^0)\big(1-2S(\widetilde{R}_{k_1-1}^0P_{-n})\big)\Big)H\big(\tilde{Q}_{k_1}^0,\widetilde{R}_{k_1-1}^0 \big)\Big\rangle_{\tilde{\eta}}=0,\\
&\Big\langle H\big(\tilde{Q}_{k_1}^0,\widetilde{R}_{k_1-1}^0 \big)^2\Big\rangle_{\tilde{\eta}}\le C \big\langle \tilde{v}_{k_1}(0)^2\big\rangle_{\tilde{\eta}_{k_1}}+
C \big\langle \tilde{u}_{k_1}(0)^2\big\rangle_{\tilde{\eta}_{k_1}}\le C_1/W^2,\\
&\Big\langle
S(\tilde{Q}_{k_1}^0)^2\big(1-2S(\widetilde{R}_{k_1-1}^0P_{-n})\big)^2\Big\rangle_{\tilde{\eta}}\le \Big\langle
S(\tilde{Q}_{k_1}^0)\Big\rangle_{\tilde{\eta}_{k_1}}\le C/W^2,\\
&\Big|\Big\langle
S(\tilde{Q}_{k_1}^0)\Big(S(\widetilde{R}_{k_1-1}^0P_{-n})-S(P_{-n})\Big)
\Big(1-2S(\widetilde{R}_{k_1-1}^0P_{-n})\Big)\Big\rangle_{\tilde{\eta}}\Big|\le \Big\langle
S(\tilde{Q}_{k_1}^0)\Big\rangle_{\tilde{\eta}_{k_1}}\le C/W^2.
\end{eqnarray*}
This,  (\ref{b1}) and (\ref{b2}) yield
\begin{multline}\label{b3}
\Big\langle
\Big( S(\widetilde{R}_{k_1}^0P_{-n})-S(P_{-n})\Big)^2
\Big\rangle_{\tilde{\eta}}\\
\le  \Big( S(\widetilde{R}_{k_1-1}^0P_{-n})-S(P_{-n})\Big)^2
\Big\rangle_{\tilde{\eta}}+C/W^2\le \ldots\le CN/W^2=o(1).
\end{multline}
Now (\ref{phi_0}) and (\ref{b3}) give (\ref{phi0}).

Therefore,
\begin{multline*}
\sum\limits_{p=1}^{Cn/W}\dfrac{(C_1)^p}{p!N^p}\sum\limits_{k_1,\ldots,k_p} \Big|\Big\langle
\widetilde{\Phi}_{k_1,\ldots,k_p}(0,0)\cdot\displaystyle\Pi_1\cdot\Pi_2\cdot\Pi_3\Big\rangle_0\Big|\\
\le \sqrt{\dfrac{CN}{W^2}}\sum\limits_{p=1}^{Cn/W}\dfrac{(C_1)^p(1+\delta)^p}{p!}\le \sqrt{C_1N/W^2}=o(1),
\end{multline*}
which together with (\ref{phi_bound}) completes the proof of Lemma \ref{l:un}.
\end{proof}
Thus, we can change $F(\overline{\vphantom{b} a},\overline{b},Q)$ to $F(0,0,I)$ in (\ref{F_0_2}), and
then integrate over
$\eta$, according to
(\ref{int_V}). We obtain
\begin{align}\notag
&\Sigma_{\pm}=2^{N}6^{N-1}W^{-6N+4}e^{-2Nc_0}
\int\limits_{\mathring{Sp}(2)}\,\,\,\int\limits_{|\widetilde{\vphantom{b}a}_j|, |\tilde{b}_j|\le W^{1-\kappa}} \mu_{c_+}(a)\,\mu_{c_-}(b)\\ \label{F_0_3}
&\times
\exp\Big\{-\sum\limits_{j=-n}^n\varphi_+(\tilde{a\vphantom{b}}_j/W)-
\sum\limits_{j=-n}^n\varphi_-(\tilde{b}_j/W)\Big\} \prod\limits_{j=-n+1}^n\Big(1-\dfrac{2}{W^2\Delta_j\Delta_{j-1}}\Big)
\\ \notag
&\times e^{-\frac{i}{2\rho(\lambda_0)}\,\Tr P_{-n}^*L_4P_{-n}\hat{\xi}_4}\,\,
\Delta_{-n}^2\Delta_n^2d\,\nu(P_{-n})\prod\limits_{q=-n}^nd \tilde{a\vphantom{b}}_q\,
d \tilde{b}_q (1+o(1))
\end{align}
Integrating over $P_{-n}$ by the Itsykson-Zuber formula (see Proposition \ref{p:Its-Z}) and
using Lemma \ref{l:maj}, we get finally
\begin{align}\notag
&\Sigma_{\pm}=\dfrac{2^{N} 6^{N-1}e^{-2Nc_0}\cdot DS(\pi(\xi_1-\xi_2))}{W^{6N-4}}\int
\limits_{|\widetilde{\vphantom{b}a}_j|, |\tilde{b}_j|\le W^{1-\kappa}} \prod\limits_{q=-n}^nd \tilde{a\vphantom{b}}_q\,
d \tilde{b}_q \cdot
 \mu_{c_+}(a)\,\mu_{c_-}(b)\\ \label{F_0_4}
&\times
(a_+-a_-+(\tilde{a\vphantom{b}}_{-n}-\tilde{b}_{-n})/W)^2(a_+-a_-+(\tilde{a\vphantom{b}}_n-\tilde{b}_n)/W)^2(1+o(1))\\
\notag
&=\dfrac{8\pi^4 \rho(\lambda_0)^4e^{-2Nc_0}(24\pi)^{N} \cdot DS(\pi(\xi_1-\xi_2))}{3\, W^{6N-4}}\,
 \Big|\mdet^{-1}\Big(-\Delta+\frac{2c_+}{W^2}\Big)\Big| (1+o(1)).
\end{align}
\end{proof}

\subsubsection{$\Sigma_+$ and $\Sigma_-$.}
In this section we prove that the integrals $\Sigma_+$ and $\Sigma_-$ over $\Omega_\delta^+$ and
$\Omega_\delta^-$ have smaller orders than $\Sigma_\pm$.

\begin{lemma}\label{l:sig_+}
For the integral $\Sigma_+$ over the domain $\Omega_\delta^+$ of (\ref{Om_delta}) we have, as $W\to\infty$
\begin{align*}\notag
|\Sigma_{+}|\le C \,W^{-2} |\Sigma_{\pm}|.
\end{align*}
The same is valid for the integral $\Sigma_-$ over the domain $\Omega_\delta^-$.
\end{lemma}
\begin{proof}
Consider $\Omega_\delta^+$ ($\Omega_\delta^-$ is similar).
Returning to $x_j$, $y_j$, $w_{j1}$, $w_{j2}$ coordinates (see (\ref{F_j})), we can write
that $\Omega_\delta^+$ corresponds to the set
\[
\widetilde{\Omega}_\delta^+=\big\{x_j, y_j, w_{j1}, w_{j2}: x_j, y_j\in U_\delta(a_+), |w_{j1}|\le \delta,
|w_{j2}|\le \delta\big\}.
\]
Change variables as
\begin{align*}
x_j&=a_++\dfrac{\tilde{x}_j}{W}, \quad w_{j1}=\dfrac{\tilde{w}_{j1}}{W},\\
y_j&=a_++\dfrac{\tilde{y}_j}{W}, \quad w_{j2}=\dfrac{\tilde{w}_{j2}}{W}.
\end{align*}
This yields
\begin{align*}\notag
&\Sigma_+=\dfrac{12^NC(\xi)^{-1}}{\pi^{2N}W^{6N}}\int\limits_{|\tilde{x}_j|,|\tilde{y}_j|\le W^{1-\kappa}} d\tilde{x}\,
d\tilde{y}
\int\limits_{|\tilde{w}_{j1}|,|\tilde{w}_{j2}|\le W^{1-\kappa}}d\Re\tilde{w}_{1}\,d\Im\tilde{w}_{1}\,
d\Re\tilde{w}_{2}\,d\Im\tilde{w}_{2}\\ \notag
&\times\exp\Big\{-\sum\limits_{j=-n+1}^n\big((\tilde{x}_j-\tilde{x}_{j-1})^2/2+
(\tilde{y}_j-\tilde{y}_{j-1})^2/2+|\tilde{w}_{j1}-\tilde{w}_{1,j-1}|^2+
|\tilde{w}_{j2}-\tilde{w}_{2,j-1}|^2\big)\Big\}\\ \notag
&\times\exp\Big\{-\dfrac{1}{2}\sum\limits_{j=-n}^n\Big(\big(a_++\dfrac{\tilde{x}_j}{W}+\dfrac{i\lambda_0}{2}+
\dfrac{i\xi_1}{N\rho(\lambda_0)}\big)^2+\big(a_++\dfrac{\tilde{y}_j}{W}+\dfrac{i\lambda_0}{2}+
\dfrac{i\xi_2}{N\rho(\lambda_0)}\big)^2\Big)\Big\}\\
&\times \exp\Big\{-\sum\limits_{j=-n}^n\big(|\tilde{w}_{j1}|^2/W^2+|\tilde{w}_{j2}|^2/W^2\big)\Big\}\\ \notag
&\times\prod\limits_{j=-n}^n\Big(\big(a_++\dfrac{\tilde{x}_j}{W}-\dfrac{i\lambda_0}{2}\big)
\big(a_++\dfrac{\tilde{y}_j}{W}-\dfrac{i\lambda_0}{2}\big)-\dfrac{|\tilde{w}_{j1}|^2+
|\tilde{w}_{j2}|^2}{W^2}\Big),
\end{align*}
which gives after some transformations
\begin{align*}
&\Sigma_+=\dfrac{12^N e^{-i\pi(\xi_1+\xi_2)}}{\pi^{2N}W^{6N}e^{2Nc_0}}\int\limits_{|\tilde{x}_j|,|\tilde{y}_j|\le W^{1-\kappa}} d\tilde{x}\,
d\tilde{y}
\int\limits_{|\tilde{w}_{j1}|,|\tilde{w}_{j2}|\le W^{1-\kappa}}d\Re\tilde{w}_{1}\,d\Im\tilde{w}_{1}\,
d\Re\tilde{w}_{2}\,d\Im\tilde{w}_{2}\\
&\times\mu_{c_+}(\tilde{x}) \cdot \mu_{c_+}(\tilde{y})\cdot\mu_{c_+}(\sqrt{2}\Re \tilde{w}_1)\cdot \mu_{c_+}(\sqrt{2}\Im \tilde{w}_1)\cdot
\mu_{c_+}(\sqrt{2}\Re \tilde{w}_2)\cdot \mu_{c_+}(\sqrt{2}\Im \tilde{w}_2) \\
&\times\exp\Big\{-\sum\limits_{j=-n}^n\Big(\dfrac{i\pi\xi_1}{N}\cdot \dfrac{\tilde{x}_j}{W}+
\phi_+(\tilde{x}_j/W)+\dfrac{i\pi\xi_2}{N}\cdot \dfrac{\tilde{y}_j}{W}+\phi_+(\tilde{y}_j/W)\Big)\Big\}\\
&\times\exp\Big\{\sum\limits_{j=-n}^n\Phi_+(\tilde{x}_j/W,\tilde{y}_j/W,\tilde{w}_{j1}/W,
\tilde{w}_{j2}/W)\Big\},
\end{align*}
where $\tilde{a}_+=a_+-i\lambda_0/2$ and
\begin{align}\label{Phi_+}
\Phi_+(x,y,w_1,w_2)=\log\Big(1-\dfrac{|w_1|^2+|w_2|^2}{(x+\tilde{a}_+)(y+\tilde{a}_+)}\Big)+
\dfrac{|w_1|^2+|w_2|^2}{\tilde{a}_+^2}.
\end{align}
Set
\begin{align*}
d\tilde{\mu}_{\gamma}=&\mu_{c_+}(\tilde{x})\,\mu_{\gamma}(\tilde{y})\,
\mu_{\gamma}(\sqrt{2}\Re \tilde{w}_1)\,
\mu_{\gamma}(\sqrt{2}\Im \tilde{w}_1)\\
&\times\mu_{\gamma}(\sqrt{2}\Re \tilde{w}_2)\, \mu_{\gamma}(\sqrt{2}\Im \tilde{w}_2)\,
d\Re\tilde{w}_{1}\,d\Im\tilde{w}_{1}\,
d\Re\tilde{w}_{2}\,d\Im\tilde{w}_{2},
\end{align*}
and let $\langle\ldots\rangle_{\tilde{\mu}_{\gamma}}$ and $\langle\ldots\rangle_{0,\tilde{\mu}_{\gamma}}$
be an expectation with respect to $d\tilde{\mu}_{\gamma}$ over $\mathbb{R}^{6N}$ or over
$[-W^{1-\kappa},W^{1-\kappa}]^{6N}$ respectively. Computing the integral $\int d\tilde{\mu}_{c_+}$ we get
\begin{align*}
&\Sigma_+=\dfrac{(24\pi)^N e^{-i\pi(\xi_1+\xi_2)}\mdet^{-3} D}{W^{6N}e^{2Nc_0}}\Big\langle
\hbox{Prod}_1(x)\cdot \hbox{Prod}_2(y)\cdot \hbox{Prod}_3 \Big\rangle_{0,\tilde{\mu}_{c_+}},
\end{align*}
where $\hbox{Prod}_{l}(x)$ and $\hbox{Prod}_3$ are the products of Taylor's series of
$\exp\{-i\pi\xi_l\tilde{x}_j/(N W)-
\phi_+(\tilde{x}_j/W)\}$, $l=1,2$ and $\exp\{\Phi_+\}$ respectively, and
\begin{align*}
D=-\Delta+\dfrac{2c_+}{W^2}.
\end{align*}
Since according to Lemma \ref{l:maj} we have
\[
\Big\langle
\hbox{Prod}_1(x)\cdot \hbox{Prod}_2(y) \Big\rangle_{0,\tilde{\mu}_{c_+}}=1+o(1),
\]
and (see Lemma \ref{l:okr})
\[
\mdet^{-1}D\le CW,
\]
we are left to prove that
\begin{equation}\label{exp_Phi}
\Big\langle
\hbox{Prod}_1(x)\cdot \hbox{Prod}_2(y)\cdot \big(\hbox{Prod}_3-1\big) \Big\rangle_{0,\tilde{\mu}_{c_+}}=o(1).
\end{equation}
Note that the series for $\exp\{\Phi_+\}$ starts from the third order. Therefore, repeating almost
literally the proof of Lemma  5 of \cite{TSh:14}, we can prove that
\[
\Big\langle\exp\Big\{\sum\limits_{j=-n}^n\Phi_+(\tilde{x}_j/W,\tilde{y}_j/W,\tilde{w}_{j1}/W,
\tilde{w}_{j2}/W)\Big\}-1\Big\rangle_{0,\tilde{\mu}_{c_+}}=o(1).
\]
The key point of Lemma  5 of \cite{TSh:14} was Lemma 6. 
The only difference in the proof of  Lemma 6 of \cite{TSh:14} is that now $g$ is a polynomial of all variables together
$\tilde{x}_j, \tilde{y}_j$, $\Re \tilde{w}_{j1}, \Im \tilde{w}_{j1}$,
$\Re \tilde{w}_{j2}, \Im \tilde{w}_{j2}$. But again we can change $\langle\ldots \rangle_{0,*}$ to $\langle\ldots \rangle_{*}$,
then write
\begin{multline*}
\Big\langle\exp\Big\{\sum\limits_j g(\tilde{x}_j, \tilde{y}_j,\tilde{w}_{j1}, \tilde{w}_{j2} )\Big\}\Big\rangle_* -1\\ \le \sum\limits_{i_1}
\Big\langle g(\tilde{x}_{i_1}, \tilde{y}_{i_1},\tilde{w}_{i_11}, \tilde{w}_{i_12} )\cdot\exp\Big\{\sum\limits_j g(\tilde{x}_j, \tilde{y}_j,\tilde{w}_{j1}, \tilde{w}_{j2} )\Big\}\Big\rangle_* ,
\end{multline*}
and apply the Wick theorem until we get $\langle\exp\Big\{\sum\limits_j g(\tilde{x}_j, \tilde{y}_j,\tilde{w}_{j1}, \tilde{w}_{j2} )\Big\}\rangle_*$ or until the number of steps
become bigger then $s_\kappa$, which is the number such that $W^{-\kappa s_\kappa}\le W^{-2}$ (see Step 2).
Now we are integrating
over $d\tilde{\mu}_{\Re c_+}$, i.e., over all variables together, and so each vertex of the
multigraph $H$ corresponding to some site $j$ consists of six parts coming from the degree of
each variables $\tilde{x}_j, \tilde{y}_j$, $\Re \tilde{w}_{j1}, \Im \tilde{w}_{j1}$,
$\Re \tilde{w}_{j2}, \Im \tilde{w}_{j2}$.  This means that some pairing are forbidden (for example,
between vertices corresponding to $(\Re \tilde{w}_{1i_1})^2\tilde{x}_{i_1}$ and
$(\Im \tilde{w}_{2i_2})^2\tilde{y}_{i_2}$), and some different pairing can correspond to the same multigraph,
but since the number of such pairing is finite (since we make the finite number of steps), it does
not change the proof (recall that matrix $M_*=-\triangle+\Re \gamma/W^2$ are the same for each set of variables
$\{\tilde{x}_j\}, \{\tilde{y}_j\}$, $\{\Re \tilde{w}_{j1}\}, \{\Im \tilde{w}_{j1}\}$,
$\{\Re \tilde{w}_{j2}\},\{ \Im \tilde{w}_{j2}\}$).

To derive  Lemma 5 of \cite{TSh:14} from Lemma 6, we should change $|x_j/W|^3$ in the bound of each addition of $\Sigma_k^0$
to $|s(w)_j^2x_j/W^3|$ or $|s(w)_j^2y_j/W^3|$, where $s(w)_j= \Re w_{j1}$, $\Im w_{j1}$, $\Re w_{j2}$ or
$\Im w_{j2}$ (note that each summand in the Taylor's series of $\exp\{\Phi_+\}$ has $s(w)_j^2/W^2$ and
$x_j/W$ or $y_j/W$), and use
\[
|s(w)^2x/W^3|\le \dfrac{p^{-1}x^2+ps(w)^4/W^4}{2}
\]
instead of  
\begin{equation*}
|x|^3\le \dfrac{p^{-1}x^2+px^4}{2}
\end{equation*}
(see eq. (4.23) in \cite{TSh:14}).

Then using Lemma \ref{l:maj} we can prove (\ref{exp_Phi}), thus Lemma \ref{l:sig_+}.
\end{proof}
This together with Lemma \ref{l:sig} yield Lemma \ref{l:sigma}.

\section{Auxiliary result}

\textbf{Proof of the Proposition \ref{p:Its-Z}.}
Statement (i) is the well-known Harish Chandra/Itsykson-Zuber formula. Its proof
can be found , e.g., in \cite{Me:91}, Appendix 5.

To prove (\ref{Its-Zub_s}) note that one can diagonalize $X$ by unitary transformation  and keep $Z$ and $T$ fixed.
Indeed, consider any unitary matrix $U$ which diagonalize $X$. Since $U\in U(2)$, it has
the form
\begin{equation}\label{u_2}
U=\left(\begin{array}{cc}
\cos\varphi\cdot e^{i\theta_1}&\sin\varphi\cdot e^{i\theta_2}\\
-\sin\varphi\cdot e^{i\theta_3}&\cos\varphi\cdot e^{i(\theta_2+\theta_3-\theta_1)}
\end{array}\right).
\end{equation}
Moreover, we can shift $U$ by any diagonal unitary matrix $U_1$. Choose $U_1$ such that
\[
U_0=UU_1=\left(\begin{array}{cc}
\cos\varphi&\sin\varphi\cdot e^{i\alpha}\\
-\sin\varphi\cdot e^{-i\alpha}&\cos\varphi
\end{array}\right).
\]
Then
\[
U_0 \sigma U_0^t=\sigma,
\]
and thus
\[
\left(\begin{array}{cc}
U_0&0\\
0&\overline{U}_0
\end{array}\right)\, F \,\left(\begin{array}{cc}
U_0&0\\
0&\overline{U}_0
\end{array}\right)^*=
\left(\begin{array}{cc}
U_0XU_0^*&w_2\, U_0\sigma U_0^t\\
-\overline{w}_2\overline{U}_0\sigma U_0^*&\overline{U}_0X^tU_0^t
\end{array}\right)=
\left(\begin{array}{cc}
\hat{X}&w_2\sigma\\
-\overline{w}_2\sigma&\hat{X}
\end{array}\right).
\]
Hence, changing $X\to U_0^*\hat{X}U_0$ (the Jacobian is $\pi/2(x_1-x_2)^2$) and using (i), we obtain
\begin{align*}
&I_t(G)=\dfrac{\pi}{2}\int_{\Omega_Y}\int_{U(2)} e^{-\frac{t}{2}\Tr(U_0^*\hat{X}U_0-D)^2-t|w_2|^2}(x_1-x_2)^2
\Phi(\hat{X},w_2)d\hat{X}\,dw_2\,d\overline{w}_2\,d\mu(U_0)\\
&=\dfrac{\pi}{2t}\int_{\Omega_Y} e^{-\frac{t}{2}\Tr(\hat{X}-D)^2-t|w_2|^2}\cdot\dfrac{x_1-x_2}{d_1-d_2}\cdot \left(1-
e^{-t(x_1-x_2)(d_1-d_2)}\right)
\Phi(\hat{X},w_2)d\hat{X}\,dw_2\,d\overline{w}_2\\
&=\dfrac{\pi}{2t}\int_{\Omega_Y} e^{-\frac{t}{2}\Tr(Y-D)^2}\cdot\dfrac{\Tr YL}{d_1-d_2}\cdot \left(1-
e^{-t\cdot\Tr YL\cdot(d_1-d_2)}\right)
\Phi(Y)dY,
\end{align*}
where
\[
Y=\left(\begin{array}{cc}
x_1&w_2\\
\overline{w}_2&x_2
\end{array}\right), \quad \mathcal{L}=\left(\begin{array}{cc}
1&0\\
0&-1
\end{array}\right),\quad dY=d\hat{X}\,dw_2\,d\overline{w}_2,\quad \Omega_Y=\{Y:F\in \Omega\}.
\]
Now diagonalizing $Y$ by the unitary transformation $V$, writing
\[
\Tr V^*\hat{Y}V\mathcal{L}=(y_1-y_2)(1-2|V_{12}|^2)
\]
and again using (\ref{Its-Zub}), we get finally
\begin{align*}
&I_t(G)=\dfrac{\pi^2}{4t}\int_{\hat{\Omega}}\int_{U(2)} \exp\Big\{-\frac{t}{2}\,
\Tr(V^*\hat{Y}V-D)^2\Big\}\cdot\dfrac{1-2|V_{12}|^2}{d_1-d_2}
\\ &\times
\left(1-
\exp\Big\{-t\,\Tr V^*\hat{Y}V\mathcal{L}\cdot(d_1-d_2)\Big\}\right) (y_1-y_2)^3dy_1\,dy_2\,d\mu(V)\\
&=\dfrac{\pi^2}{4t^2}\int dy_1\,dy_2\, \exp\Big\{-\frac{t}{2}\,\Tr(\hat{Y}^2+D^2)\Big\}
\cdot\Phi(\hat{Y})\cdot \dfrac{(y_1-y_2)^2}
{(d_1-d_2)^2}\\
&\times \Big[e^{t(y_1d_1+y_2d_2)}\cdot \Big(1-\dfrac{2}{t(y_1-y_2)(d_1-d_2)}\Big)+
e^{t(y_1d_2+y_2d_1)}\cdot \Big(1+\dfrac{2}{t(y_1-y_2)(d_1-d_2)}\Big)\Big],
\end{align*}
which, taking into account the symmetry of $\hat{\Omega}$, yields (\ref{Its-Zub_s}). Integral
(\ref{Its-Zub_s1}) can be computed straightforward. $\quad \Box$
 \medskip

\textbf{Proof of Lemma \ref{l:phi}.}
Note that all non-zero moments of measure $d\eta$ can be expressed via
expectations of $v_j^{2s}:=|(V_j)_{12}|^{2s}$, $u_j^{2l}:=|(U_j)_{12}|^{2l}$.
In addition, according to Proposition~ \ref{p:Its-Z},
\begin{align*}
&\langle v_j^{2s}u_j^{2l}\rangle_{\eta_j}=12q_j^{-1}\int\limits_{0}^1
 v_j^{2s+1}u_j^{2l+1}e^{t_j(1-2v_j^2)(1-2u_j^2)/2-t_j/2}(1-2v_j^2)^2du_j\, dv_j\\
&= 24q_j^{-1}\int\limits_{0}^1du_j\int\limits_{0}^{1/\sqrt{2}}dv_j\,\,
 v_j^{2s+1}u_j^{2l+1}e^{t_j(1-2v_j^2)(1-2u_j^2)/2-t_j/2}(1-2v_j^2)^2\\
&=\dfrac{24q_j^{-1}}{W^4p_j^4}\int\limits_{0}^{p_jW/\sqrt{2}}v_j\, dv_j
\int\limits_{0}^{p_jW\sqrt{1-2v_j^2/p_j^2}}u_j\,du_j\,\,\Big(\dfrac{v_j}{Wp_j}\Big)^{2s}
\cdot\Big(\dfrac{u_j}{Wp_j(1-2v_j^2/W^2p_j^2)^{1/2}}\Big)^{2l}\\
&\times \exp\big\{-(a_+-a_-)^2(v_j^2+u_j^2)\big\}\cdot\Big(1-\dfrac{2v_j^2}{W^2p_j^2}\Big)\\
&=\langle\tilde{v}_j^{2s}\tilde{u}_j^{2l}\cdot\big(1-2\tilde{v}_j^2\big)\rangle_{\tilde{\eta}_j}
+O(e^{-C_1W^2}),
\end{align*}
where $\eta_j$, $q_j$ and $t_j$ are defined in (\ref{F_V}), and in the third line
we have changed $t_jv_j^2\to (a_+-a_-)^2v_j^2$, $t_j(1-2v_j^2)u_j^2\to (a_+-a_-)^2u_j^2$.

Now let $\mathbf{E}_k$ be the averaging with respect to the product of the measures
$d\tilde{\eta}_j$ for $j$ from $(-n+1)$ to $(-n+k)$ and the measures $d\eta_j$ for
$j$ from $(-n+k+1)$ to $n$. Thus, if
\[\Psi_{k_1,\ldots,k_s}=\prod_{j=1}^s S(R_{k_j}P_{-n}), \]
then it suffices to estimate
\[
\Big|\tilde{\Psi}_{k_1,\ldots,k_s}^0-\tilde{\Psi}_{k_1,\ldots,k_s}^{2n}\Big|\le e^{-cW^2}
\]
for $s\le p$, where
\[
\tilde{\Psi}_{k_1,\ldots,k_s}^i=\mathbf{E}_{i}\Big\{\Psi_{k_1,\ldots,k_s}
\prod\limits_{j=-n+1}^{-n+i}(1-2\tilde{v}_j^2)\Big\}.
\]
Note that
\[
\Big|\tilde{\Psi}_{k_1,\ldots,k_s}^0-\tilde{\Psi}_{k_1,\ldots,k_s}^{2n}\Big|
\le\sum\limits_{i=1}^{2n}\Big|\tilde{\Psi}_{k_1,\ldots,k_s}^{i-1}-\tilde{\Psi}_{k_1,\ldots,k_s}^{i}\Big|.
\]
In each summand we write for $\gamma=i-1,i$ (we assume that all $k_j\ge (-n+i)$)
\begin{align*}
&\tilde{\Psi}_{k_1,\ldots,k_s}^\gamma
=\mathbf{E}_{\gamma}\Big\{\prod_{j=1}^s
S\big(R_{-n+i-1}Q_{-n+i}(R_{-n+i}^*R_{k_j}P_{-n})\big)\prod\limits_{j=-n+1}^{-n+\gamma}
(1-2\tilde{v}_j^2)\Big\}\\&=
\mathbf{E}_{\gamma}\Big\{\prod_{j=1}^s \sum\limits_{l=2,4}|\sum_{\alpha,\alpha'=1,\ldots,4}
(R_{-n+i-1})_{1\alpha}(Q_{-n+i})_{\alpha\alpha'}
(R_{-n+i}^*R_{k_j}P_{-n}))_{\alpha'l}|^{2}\prod\limits_{j=-n+1}^{-n+\gamma}
(1-2\tilde{v}_j^2)\Big\}\\&
=\left\{
\begin{array}{ll}
\sum_{k,l=1}^{s+1} C_{k,l}\mathbf{E}_{\gamma}\{|(V_{-n+i})_{12}|^{2k}|(U_{-n+i})_{12}|^{2l}\},\quad \gamma=i-1,\\
\sum_{k,l=1}^{s+1} C_{k,l}\mathbf{E}_{\gamma}\{|(V_{-n+i})_{12}|^{2k}|(U_{-n+i})_{12}|^{2l}(1-2\tilde{v}_{-n+i}^2)\},\quad \gamma=i,
\end{array}\right.
\end{align*}
where the coefficients $C_{k,l}$ are the same for $\gamma=i$ and $\gamma=i-1$ and can be bounded by $C^{s}$,
since $|(R_{-n+i-1})_{1\alpha}|\le 1$ and $|(R_{-n+i}^*R_{k_j}P_{-n}))_{\alpha'l}|\le 1$, $l=2,4$.
Moreover, since
\begin{multline*}
|\mathbf{E}_{i-1}\{|(V_{-n+i})_{12}|^{2k}|(U_{-n+i})_{12}|^{2l}\}\\-
\mathbf{E}_{i}\{|(V_{-n+i})_{12}|^{2k}|(U_{-n+i})_{12}|^{2l}(1-2\tilde{v}_{-n+i}^2)\}|\le C^s k!l! e^{-CW^2},
\end{multline*}
we obtain
\[\Big|\tilde{\Psi}_{k_1,\ldots,k_s}^{0}-\tilde{\Psi}_{k_1,\ldots,k_s}^{2n}\Big|\le n C^p_1(p!)^2e^{-CW^2}\le n e^{C_2 (n\log n)/W}e^{-CW^2}=O(e^{-C_2W^2}).
\]
This yields Lemma \ref{l:phi}.
$\quad \Box$

{\bf Acknowledgements.} This research was sponsored by the grant of the Russian
Science Foundation (project 14-21-00035)

\end{document}